\documentclass[11pt,a4paper,reqno]{amsart}%
\usepackage{amsthm,amsmath,amsfonts,amssymb,amsxtra,appendix,bookmark,euscript,delarray,dsfont,bm,mathrsfs}
\usepackage[latin1]{inputenc}
\usepackage{color}

\theoremstyle{plain}
\newtheorem{theorem}{Theorem}
\newtheorem{lemma}[theorem]{Lemma}
\newtheorem{corollary}[theorem]{Corollary}
\newtheorem{proposition}[theorem]{Proposition}

\theoremstyle{remark}
\newtheorem{remark}[theorem]{Remark}
\newtheorem{definition}[theorem]{Definition}

\renewcommand\phi{\varphi}

\renewcommand{\epsilon}{\varepsilon}

\renewcommand{\ge}{\geqslant}
\renewcommand{\le}{\leqslant}
\renewcommand{\geq}{\geqslant}
\renewcommand{\leq}{\leqslant}
\renewcommand{\hat}{\widehat}
\renewcommand{\tilde}{\widetilde}

\allowdisplaybreaks

\title[On the energy cascade  of 3-wave kinetic equations]{On the energy cascade  of 3-wave kinetic equations: Beyond Kolmogorov-Zakharov solutions}

\author[A. Soffer]{Avy Soffer}
\address{Mathematics Department, Rutgers University, New Brunswick, NJ 08903 USA.} 
\email{soffer@math.rutgers.edu}

\author[M.-B. Tran]{Minh-Binh Tran}
\address{Department of Mathematics, Southern Methodist University, Dallas, TX 75275, USA} 
\email{minhbinht@mail.smu.edu}

\begin{document}

\begin{abstract} 
In weak turbulence theory, the Kolmogorov-Zakharov spectra is a class of time-independent solutions to the kinetic wave equations. In this paper, we construct a new class of time-dependent isotropic solutions to {the decaying turbulence problems (whose solutions are energy conserved), with general initial conditions}. These solutions exhibit the  interesting property that the energy is cascaded from small wavenumbers to large wavenumbers. We can prove that starting with a regular initial condition whose energy at the infinity wave number $|p|=\infty$  is $0$, as time evolves, the energy is gradually accumulated at $\{|p|=\infty\}$. Finally, all the energy of the system is concentrated at $\{|p|=\infty\}$ and the energy function becomes a Dirac function at infinity $E\delta_{\{|p|=\infty\}}$, where $E$ is the total energy. The existence of this class of solutions is, in some sense, {the first complete rigorous} mathematical proof based on the kinetic description for the  energy cascade phenomenon for waves with quadratic nonlinearities. We only represent in this paper the analysis of the statistical description of  acoustic waves (and equivalently capillary waves). However, our analysis works for other cases as well.  
\end{abstract}

%
%
\maketitle

\tableofcontents

\section{Introduction}


Over more than half a century, the theory of weak wave turbulence   has been intensively developed. In weakly nonlinear and dispersive wave models, the weak turbulence kinetic equations can be formally derived, via the statistical approach, to describe the dynamics of resonant wave interactions. The  ideas of deriving the kinetic equations to describe how energies are shared between weakly interacting waves go back to Peierls \cite{Peierls:1993:BRK,Peierls:1960:QTS} and the modern devevelopments have the origin in the works of Hasselman \cite{hasselmann1962non,hasselmann1974spectral}, Benney and Saffmann \cite{benney1966nonlinear}, Kadomtsev \cite{kadomtsev1965plasma}, Zakharov \cite{zakharov2012kolmogorov},  Benney and Newell \cite{benney1969random}, Lukkarinen and Spohn  \cite{lukkarinen2009not,LukkarinenSpohn:WNS:2011,Spohn:TPB:2006}.  We refer to the books \cite{Nazarenko:2011:WT,zakharov2012kolmogorov} for more discussions and references on the topic.

One of the most important wave turbulence  equations is the so-called 3-wave kinetic equation (cf. \cite{pushkarev1996turbulence,pushkarev2000turbulence,zakharov1968stability,zakharov1967weak,ZakharovNazarenko:DOT:2005})
\begin{equation}\label{WeakTurbulenceInitial}
\begin{aligned}
\partial_tf(t,p) \ =& \ \mathbb{Q}[f](t,p), \\\
f(0,p) \ =& \ f_0(p),
\end{aligned}
\end{equation}
in which  $f(t,p)$ is the nonnegative wave density  at  wavenumber $p\in \mathbb{R}^N$, $N \ge 2$; $f_0(p)$ is the initial condition.  The quantity $\mathbb{Q}[f]$ denotes the integral collision operator, describing pure resonant
3-wave interactions. In the  3-wave turbulence kinetic equation, the collision operator is of the form 
\begin{equation}\label{def-Qf}{Q}[f](p) \ = \ \iint_{\mathbb{R}^{2N}} \Big[ R_{p,p_1,p_2}[f] - R_{p_1,p,p_2}[f] - R_{p_2,p,p_1}[f] \Big] d^Np_1d^Np_2 \end{equation}
with $$\begin{aligned}
R_{p,p_1,p_2} [f]:=  |V_{p,p_1,p_2}|^2\delta(p-p_1-p_2)\delta(\omega -\omega_{1}-\omega_{2})(f_1f_2-ff_1-ff_2) 
\end{aligned}
$$
with the short-hand notation $f = f(t,p)$, $\omega = \omega(p)$ and $f_j = f(t,p_j),$ $\omega_j = \omega(p_j)$, for wavenumbers $p$, $p_j$, $j\in\{1,2\}$. The quantity $\omega(p)$ denotes the dispersion relation of the waves. The equation describes, under the assumption of weak nonlinearities, the spectral energy transferred on the resonant manifold, which is a set of wave vectors satisfying
\begin{equation}\label{cv} p = p_1 + p_2 , \qquad \omega =\omega_{1} + \omega_{2}.\end{equation}
The exact form of the collision kernel $V_{p,p_1,p_2}$  depends on the type of waves under consideration. The 3-wave  kinetic equation plays a very important role in the theory of weak turbulence, with a variety of applications for ocean waves, acoustic waves, gravity capillary waves, Bose-Einstein condensate and many others (see \cite{hasselmann1962non,hasselmann1974spectral,zakharov1965weak,zakharov1968stability,zakharov1967weak,KorotkevichDyachenkoZakharov:2016:NSO,zakharov1999statistical,
 connaughton2009numerical,zakharov1967weak,newell2011wave,balk1990physical,
 pushkarev1996turbulence,pushkarev2000turbulence,l1997statistical} and references therein).

One of the most important results of the weak turbulence theory (cf. \cite{KorotkevichDyachenkoZakharov:2016:NSO,zakharov1967weak,zakharov2012kolmogorov}) is  the existence of the so-called  Kolmogorov-Zakharov spectra, which is a class of {\it time-independent solutions} $f_\infty$ of equation \eqref{WeakTurbulenceInitial}:  $$f_\infty(p) \approx C|p|^{-\kappa}, \ \ \ \kappa>0.$$
 These solutions are the analogs of the  Kolmogorov energy spectrum 
$C|p|^{-\frac{5}{3}}$ of hydrodynamic turbulence.  Normally, thermodynamic equilibrium solutions can be easily derived  from  kinetic equations by inspection. However,  the Kolmogorov-Zakharov (KZ) spectra are much more subtle and only emerge after one has exploited scaling symmetries of the dispersion relations and the coupling
coefficients via what is now called the Zakharov transformation. The discovery of the KZ spectra has been so far a milestone and breakthrough achievement on the subject, and for decades thereafter, it has been the dominating subject of study in the wave turbulence theory. Works on this research line have been continued till now and the Kolmogorov-Zakharov spectra have been found in new applications, including   astrophysics \cite{galtier2000weak},  interior ocean \cite{lvov2010oceanic},  cosmology \cite{micha2004turbulent} and several other physical situations.

On the other hand, relatively little is known about the time-dependent solutions of \eqref{WeakTurbulenceInitial}. In the important works \cite{connaughton2009numerical,connaughton2010aggregation,connaughton2010dynamical}, several numerical experiments were performed, to  investigate the 3-wave equation. In these works, the following equivalent form of the 3-wave collision operator was introduced 
\begin{align}\label{EE1Colm}
\begin{split}
\mathbb{{Q}}[f](t,\omega) \ =& \ \int_0^\infty\int_0^\infty \big[R(\omega, \omega_1, \omega_2)-R(\omega_1,\omega, \omega_2)-R(\omega_2, \omega_1, \omega) \big]{d}\omega_1{d}\omega_2, \\\
& R(\omega, \omega_1, \omega_2):=  \delta (\omega-\omega_1-\omega_2)
\left[ U(\omega_1,\omega_2)f_1f_2-U(\omega,\omega_1)ff_1-U(\omega,\omega_2)ff_2\right]\,,
\end{split}
\end{align}
where $U$ satisfies $| U(\omega_1,\omega_2) |  \ = \ (\omega_1\omega_2)^{\gamma/2}.$

The solutions of \cite{connaughton2010aggregation,connaughton2010dynamical} are assumed to follow the so-called {\it dynamic scaling hypothesis}
\begin{equation}\label{Ansartz}
f(t,\omega)\approx s(t)^a F\left(\frac{\omega}{s(t)}\right)
\end{equation}
 in which $\approx$ denotes the scaling limit $s(t)\to\infty$ and $\omega\to \infty$ with $x=\omega/s(t)$ fixed. Let us compute the energy of this function
\begin{equation}\label{EnergyColm}
\begin{aligned}
\int_{0}^\infty \omega f(t,\omega)d\omega\ = & \ \int_0^\infty s(t)^a F\left(\frac{\omega}{s(t)}\right)\omega d\omega \ =  \ \int_0^\infty s(t)^{a+2} F\left(\frac{\omega}{s(t)}\right)\left(\frac{\omega}{s(t)}\right) d\left(\frac{\omega}{s(t)}\right)\\
\ = & \ s(t)^{a+2} \int_0^\infty x F\left(x\right) dx,
\end{aligned}
\end{equation}
which grows with the rate $s(t)^{a+2}$.
Substituting this ansatz into the equation \eqref{WeakTurbulenceInitial}-\eqref{EE1Colm}, we get the system
\begin{equation}
\begin{aligned}
& \dot{s}(t) \ =  \ s^\zeta, \mbox{ with } \zeta = \gamma + a + 2\\
& aF(x) \ + \ x\dot{F}(x) \ =  \ \bar{\mathbb{Q}}[F](x).
\end{aligned}
\end{equation}
In \cite{connaughton2010aggregation}, the solutions follow the so-called  {\it decaying turbulence}, in which the  energy is supposed to be the same for all time \begin{equation}\label{ColmEnergy1}\int_{0}^\infty \omega f(t,\omega)d\omega=const.\end{equation}
From \eqref{EnergyColm}, it is easily seen  that the only value of $a$ that gives the conservation of energy is $a=-2$. By assuming that $F(x)\backsim x^{-n}$, when $x\backsim 0$, the power can be determined $n=\gamma+1$. Since the degree of homogeneity $\gamma$ is considered in the interval $[0,1)$, the integral $\int_0^\infty xF(x) dx$ is well-defined. However, this integral becomes singular if $\gamma>1$.

In \cite{connaughton2010dynamical}, the degree of homogeneity $\gamma$ is considered in the interval $[0,2]$. Suppose that the solutions follow the so-called {\it forced turbulence}, whose   total energy is assumed to grow linearly in time \begin{equation}\label{ColmEnergy2}
\int_{0}^\infty \omega f(t,\omega)d\omega=Jt.
\end{equation}
 {Using \eqref{ColmEnergy2}}, we obtain $$\dot{s} = \frac{J}{(a+2)\int_0^\infty xF(x)dx}s^{-1-a}$$ and then $a=-\frac{\gamma+3}{2}$.  The challenge is that the integration $\int_0^\infty xF(x) dx$ with $F(x)\backsim x^{{-\frac{\gamma+3}{2}}}$ for small $x$  diverges in the finite capacity case $\gamma>1$ and converges only in the infinite capacity case $\gamma<1$. Many  strategies have then been introduced, mainly to approximate $a$ directly from the eigenvalue problem, to overcome the challenge. In those cases, the energy grows with the rate  $s(t)^{a+2}$.

From these  numerical experiments, the dependence on $\gamma$ of the behavior of the solutions  can be clearly seen. A deeper theoretical understanding of the numerical experiments is then required.

In this work, we consider a very high value of the degree of homogeneity $\gamma$ and the physical situation when the energy of the solutions is conserved in time (the system is not driven by injected  energy.) On this decaying turbulence system, we perform the first rigorous mathematical analysis of time-dependent spatially homogeneous and isotropic   solutions, to understand their long time behavior. To be more precise, we consider the case $\gamma=2$, which corresponds to acoustic waves, with
\begin{equation}\label{EE1}
|V_{p,p_1,p_2}|^2 \ = \ |p||p_1||p_2|, \ \ \ N=3  \mbox{ and } \omega(p)=|p| .\end{equation} 
{\it Note that capillary waves also have $\gamma=2$  and our analysis can be equivalently applied to this case.} We show that energy conserved solutions of 3-wave systems in this case are only local in time. In other words, if we consider a 3-wave turbulence system, whose kernel degree of homogeneity is high ($\gamma=2$), and look for a solution whose energy is a constant for all time, then this solution can exist only up to a finite time, after that, some {\it energy is lost to infinity}. At the first sight, this kind of behavior looks surprising and mysterious. One may ask the question: ``Where does the energy go?'' We will explain in the next section that this phenomenon is closely related to {\it the gelation phenomenon} in coagulation-fragmentation models (cf. \cite{escobedo2002gelation,filbet2004numerical,leyvraz1983existence,leyvraz1981singularities}), that corresponds the formation of a ``giant particle'' with ``infinite size'' $(\omega=\infty)$ in finite time. Since all of the particles considered by  coagulation-fragmentation models
have finite size $0\le \omega<\infty$, the occurrence of gelation results in {\it  a loss of mass  to infinity}, when the degree of homogeneity $\gamma$ of the kernel of the coagulation part $(\omega_1\omega_2)^{\gamma/2}$ is large  $\gamma>1$. 


In continuum mechanics, an energy cascade is the transfer  of energy from large scales to  small scales - direct energy cascade, or the transfer of energy from  small scales to  large scales -  inverse energy cascade. Since the energy cascade phenomenon is related to the evolution in time of the solutions, a natural question  is that:``Can we observe from some time-dependent isotropic solutions of   kinetic wave equations that  the energy  is transferred to large/small values of  wavenumbers as  time goes to infinity''? Our analysis  shows that,   {time-dependent spatially homogeneous and isotropic   solutions} to the 3-wave equations with a high degree of homogeneity, in the decaying turbulence case, indeed exhibit the energy cascade phenomenon.

Our work shows that in many classical examples of waves (acoustic waves, capillary waves), different from the KZ spectra, {time-dependent isotropic and energy conserved  solutions} to the 3-wave equations follow the energy cascade phenomenon and exhibit  the energy loss as time evolves. We only present in this paper the analysis of the statistical description of  acoustic waves  $\gamma=2$ (and equivalently, capillary waves). {\it However, our analysis works for other cases as long as the growth of the kernel   $\gamma$ is strong enough $(\gamma>1)$,  in consistent with the coagulation-fragmentation  phenomena. }

{In the PhD Thesis of Kierkels (cf. \cite{kierkels2015transfer,kierkels2016self}), the following equation was studied (cf. \cite{kierkels2015transfer}[Equation (1.9)])
\begin{equation}\label{KV}
\begin{aligned}
\partial_t\left(\int_{[0,\infty)}\varphi(\omega)G(t,\omega) d\omega\right) \ = &  \ \iint_{\omega_1\ge \omega_2\ge 0}\frac{G_1G_2}{\sqrt{\omega_1\omega_2}}[\varphi(\omega_2)+\varphi(\omega_1-\omega_2)-\varphi(0)-\varphi(\omega_1)]d\omega_1 d\omega_2\\
+ &\ \iint_{[0,\infty)^2}\frac{G_1G_2}{\sqrt{\omega_1\omega_2}}[\varphi(0)+\varphi(\omega_1+\omega_2)-\varphi(\omega_1)-\varphi(\omega_2)]d\omega_1 d\omega_2,
\end{aligned}
\end{equation}
for some suitable test function $\varphi$. Equation \eqref{KV} describes the time evolution of the fraction of mass $G$ that is not supported near the origin of the 4-wave kinetic equation derived from the cubic nonlinear Schr\"odinger equation. The mass and energy of the solution are conserved in time $\partial_t\int_{[0,\infty)}  G(t,\omega) d\omega=\partial_t\int_{[0,\infty)}\omega G(t,\omega) d\omega=0$. In this interesting  work, it is proved that any nontrivial initial datum yields the instantaneous onset of a condensate and self-similar solutions are also constructed. The strategy of the proof follows an adaptation of the monotonicity estimate introduced by Lu (cf. \cite{Lu:2005:TBE,Lu:2013:TBE,Lu:2014:TBE,Lu:2018:LTS,Lu:2016:LTC}) and developed by Escobedo-Velazquez (cf. \cite{EscobedoVelazquez:2015:FTB,EscobedoVelazquez:2015:OTT}) to the case of \eqref{KV}, in which the kernel is singular at $0$, leading to the instantaneous condensation. }

{Physically speaking,  \eqref{KV} is very different from the model under investigation. We show that there is no instantaneous condensation, but rather the ``reversed'' physical phenomenon: the  energy  cascades to infinity partially in finite time and totally in infinite time. Moreover, in the model under investigation, the mass is not conserved (cf. \cite{zakharov2012kolmogorov}); while in \eqref{KV}, both the mass and energy are conserved.}

{Similar with \cite{kierkels2015transfer}, our existence proof   also  relies on the classical kernel cutting off strategy for the  coagulation-fragmentation equation  (see, for instance, \cite{dubovski1996existence}) and  the classical homogeneous Boltzmann equation (see, for instance, \cite{Arkeryd:1972:OBE}). To prove the energy cascade, we develop new monotonicity estimates for 3-wave kinetic equations,  based  on previous works \cite{EscobedoVelazquez:2015:FTB,EscobedoVelazquez:2015:OTT,kierkels2015transfer,Lu:2005:TBE,Lu:2013:TBE,Lu:2014:TBE,Lu:2018:LTS,Lu:2016:LTC}. We also refer to \cite{EscobedoVelazquez:2015:OTT} for discussions on related topics for the 4-wave equations. In this interesting work, it is shown that the condensation occurs in finite time in the form of a delta function at the origin, and the energy, therefore, goes to infinity in infinite time.}


\section{Comparisons with related models, brief descriptions of the main results,  and outline of the proof}
\subsection{Description of the problem and comparison with coagulation-fragmentation models} 
As discussed in the introduction,  if we consider a spatially homogeneous and isotropic capillary or acoustic kinetic wave equation, and look for a solution whose energy is a constant for all time, then this solution can exist only up to a finite time, after this time, some {\it energy is lost to infinity}. In order to capture the the dynamics of the energy cascade phenomenon, we rewrite a new equation for the energy $g({|p|})=\omega_{|p|} f(|p|)$. This equation becomes a sophisticated coangluation-fragmentation type equation, whose fragmentation term is nonlinear rather than linear  (see Definition \ref{Def:WeakSolution} and (R1.1)). We construct a new weak formulation of the solutions in an extended space containing the delta function $\delta_{\{|p|=\infty\}}$, whose role is to  capture the energy loss at infinity.
\begin{remark}
Note that the measure at infinity $\{|p|=\infty\}$ can be defined by using the  ``extended half real line'' $[0,\infty] \ = \ [0,\infty)\cup\{ \infty\}.$ Suppose that $\mathfrak{B}([0,\infty))$ is the set of Borel sets of $[0,\infty)$, we can define the set of Borel sets of the extended real line by
$\mathfrak{B}([0,\infty]) \ =  \ \{B\subset [0,\infty]: B\cap [0,\infty) \in \mathfrak{B}([0,\infty))\}.$   
Therefore, one can define the space of finite nonnegative  measures  in $\mathfrak{B}([0,\infty])$, including $\delta_{\{|p|=\infty\}}$. The measure $d\mu(|p|)$ is the extension of the classical Lebesgue measure in $\mathfrak{B}([0,\infty])$. This  construction is classical and could be found in textbooks in analysis, for instance \cite{folland1995introduction}.
\end{remark}
With the new weak formulation, we introduce several test functions, which are continuous function on $[0,\infty]$,  in to order to capture the behavior of the solutions at  $|p|=\infty$. Different from the KZ spectra, using these new solutions, we can explicitly see that the conserved energy is accumulated at large values of $|p|$, and finally, the energy will be concentrated at  $|p|=\infty$: they become a Dirac function $E\delta_{\{|p|=\infty\}}$ at the limit $t\to\infty$, where $E$ is the total energy of the solutions. Indeed, with these solutions, one could see that all of the energy goes to infinity and  that the energy in every given interval $0\le |p| \le R$ for any positive number $R$ vanishes as $t$ tends to infinity. 
{\it We show that there is an ``energy cascade event'', similar to the gelation event discussed in (R1.1), in which a part of the wave energy is lost  in finite time. The lost energy can be proved to accumulate into a delta function at infinity. } This is, in some sense, a rigorous proof of the energy cascade phenomenon for acoustic and capillary wave systems.

We have the following comparisons.

\begin{itemize}

\item[\bf(R1.1)] 

The operator ${{Q}}[f]$ is very similar to the Smoluchowski collision operator,  which takes the following form (cf. \cite{von1916drei})
\begin{equation}\label{Smo}\mathbb{Q}_{\text Smo}[f](\omega)={\frac  {1}{2}}\int _{0}^{\omega}U_{\text Smo}(\omega-\omega_1,\omega_1)f(\omega-\omega_1)f(\omega_1)\,d\omega_1-\int _{0}^{\infty }U_{\text Smo}(\omega,\omega_1)f(\omega)f(\omega_1)\,d\omega_1,
\end{equation}
for some kernel $U_{\text Smo}$. The integrodifferential equation that takes the Smoluchowski operator as the collision operator - the Smoluchowski equation -  describes the time evolution of the number density of particles as they coagulate  to size $\omega$ at time $t$. There has been a large body of research on the mathematical analysis of  the Smoluchowski equation (see \cite{bonacini2019self,canizo2010rate,degond2017coagulation,escobedo2006dust,filbet2004numerical,melzak1957scalar,menon2006dynamical})  
.

To describe the dynamics of the cluster growth, in which the sizes of the clusters evolve with time as the clusters undergo not only coagulation but also fragmentation events, the coagulation-fragmentation models have been introduced with an additional linear fragmentation operator (see \cite{bertoin2006random,drake1972general,escobedo2003gelation,laurencot2015absence} and references therein) 
\begin{equation}\label{Frag}\begin{aligned}
\mathbb{Q}_{\text CoFr}[f](\omega) \ = & \ \mathbb{Q}_{\text Smo}[f] \ - \ \mathbb{Q}_{\text Fr}[f],\\
\mathbb{Q}_{\text Fr}[f](\omega)\ = & \ a(\omega)f(\omega)\ - \ \int _{0}^{\infty }a(\omega)b(\omega,\omega_1)f(\omega_1)\,d\omega_1.
\end{aligned}
\end{equation}
for some kernels $a,b$.

In pure coagulation dynamics, particles coagulate into clusters, which only get larger.  On the other hand, in wave turbulence phenomena, waves can either combine with other waves to form waves with larger wavenumbers (which corresponds to the delta function $\delta(\omega-\omega_1-\omega_2)$ in the formulation of the collision operator), or break into waves with smaller wavenumbers  (which corresponds to the delta functions $\delta(\omega_1-\omega-\omega_2)$ and $\delta(\omega_2-\omega-\omega_1)$ in the formulation of the collision operator).  As a result,  the terms containing $[f(\omega)-f(\omega_1)]f(\omega_1-\omega)$, with $\omega_1\ge\omega$ are missing in the presentation of the  Smoluchowski collision operator.  In  coagulation-fragmentation dynamics, the fragmentation is described by a linear operator. The additional linear operator $\mathbb{Q}_{\text Fr}[f]$ could be seen as a linear version of the nonlinear terms containing  $\delta(\omega_1-\omega-\omega_2)$ and $\delta(\omega_2-\omega-\omega_1)$ and describing the breaking  of  waves into waves with smaller wavenumbers. During coagulation events, the total mass of particles is expected to be a constant throughout the time evolution of solutions of coagulation-fragmentation models. This is, indeed, a fundamental difference between coagulation-fragmentation models and wave turbulence models. In 3-wave turbulence models, the  mass is not expected to be conserved but the energy is. {\it If we denote the new unknown $g=f\omega$, then the 3-wave turbulence kinetic equation in $g$  becomes a sophisticated coagulation-fragmentation equation (see Definition \ref{Def:WeakSolution}), in which the fragmentation term is nonlinear, rather than linear. In this equation, the mass of $g$ (which is also the energy of $f$) is expected to be conserved.}

 When the fragmentation is absent, the coagulation kernel $U_{\text Smo}(\omega_1,\omega_2)$ can be split into two classes $0\le U_{\text Smo}(\omega_1,\omega_2)\lesssim 2+\omega_1+\omega_2$ and $U_{\text Smo}(\omega_1,\omega_2)\gtrsim (\omega_1\omega_2)^{\gamma/2}$ for some $\gamma>1$. In the first case, the solutions are expected to be mass-conserved; and in the second case, the solutions have gelation in finite time, in which the conservation of mass breaks down.  The gelation corresponds to a runaway growth of the dynamics, that leads to the formation of a giant particle with infinite size $\omega=\infty$ in finite time. Since all of the particles considered by  coagulation-fragmentation models
have finite size $0\le \omega<\infty$, the occurrence of gelation results in a loss of mass.  The gelation phenomenon was conjectured in the 80s \cite{leyvraz1983existence,leyvraz1981singularities} and its mathematical proof for  such coagulation kernels and arbitrary initial data was shown twenty years later in \cite{escobedo2002gelation}. A natural question is that, {\it after  gelation occurs, where the mass goes}. This question has not yet been answered in \cite{escobedo2002gelation}. Since the fragmentation reduces the sizes of the clusters, it was proved that a strong fragmentation prevents the occurrence of gelation \cite{costa1995existence}  but does not affect  the coagulation events if it is  weak \cite{laurenccot2000class,vigil1989stability}. In general, due to the effect of gelation, the existence analysis of coagulation-fragmentation models follows two main streams. In the first stream, several works have been devoted to the construction of mass-conserving solutions to models whose kernels satisfy $0\le U_{\text Smo}(\omega_1,\omega_2)\lesssim 2+\omega_1+\omega_2$ with  various assumptions on $a$ and $b$ (see  \cite{ball1990discrete,banasiak2011global,dubovski1996existence,lamb2004existence,banasiak2012analytic,white1980global}  and the references therein). In the second stream, weak solutions which need not satisfy the mass conservation have been constructed (cf. \cite{eibeck2001stochastic,escobedo2005self,giri2012weak,norris1999smoluchowski,stewart1989global} and the references therein). The existence of mass-conserving solutions the models  with kernels satisfying $U_{\text Smo}(\omega_1,\omega_2)\gtrsim (\omega_1\omega_2)^{\gamma/2}$ with strong fragmentations has been done in \cite{costa1995existence,escobedo2003gelation}. An opposite phenomenon takes place when the 
fragmentation rate $a$ takes the special forms $a(x)=a_0 x^\xi$ for some $\xi<0$. In this case, the smaller the particles, the faster they divide, which leads to 
the  appearance of dust and again a loss of
mass takes place. This phenomenon is usually referred to as the shattering transition \cite{arlotti2004strictly,McGrady}.   A survey of earlier results can be found in \cite{laurenccot2004coalescence}.

\item[\bf(R1.2)]  As discussed above, there are  differences between 3-wave turbulence and coagulation - fragmentation models. First,  the terms describing the breaking of waves/particles are nonlinear in the 3-wave turbulence equation while they are linear in  coagulation-fragmentation models. Second, since the mass is not conserved in 3-wave turbulence phenomena, if we denote the new unknown $g=f\omega$, then the 3-wave turbulence kinetic equation in $g$  becomes a highly sophisticated coagulation-fragmentation equation. However, there is a deep connection between the two types of models. In coagulation-fragmentation models, it is shown in \cite{escobedo2002gelation} that if the growth of the coagulation kernels is sufficiently strong $ (\omega_1\omega_2)^{\gamma/2}$ with $\gamma>1$, then a giant particle of infinite size is formed in finite time.    Since in these models, $\omega$ denotes to the size of the particle,  the giant particle of infinite size corresponds to $\omega=\infty$. In our wave turbulence model, we will show later that the conservation of energy is broken at finite time.  The energy breaking in the energy cascade phenomenon is closely related to the gelation phenomenon. {In order to describe the dynamics of this phenomenon, we construct a new weak formulation of the energy solutions $g$ in an extended space containing the delta function $\delta_{\{\omega=\infty\}}$. With this weak formulation, several test functions are introduced,  in to order to capture the behavior of the solutions at  $\omega=\infty$.}  That allows us to show that all of the energy of the system will finally go to this delta function as time goes to infinity.  {In other words, we show that there is also a similar ``gelation'' event for waves, in which some of the energy of the waves accumulate into a delta function with infinite wave number $\omega=\infty$ at finite time; and as time evolves, all of the energy of the waves will be ``gelled'' into this delta function.}  { The  delta function at infinity $\omega=\infty$ corresponds to the giant  particle with infinite size in the coagulation-fragmentation case. In this  picture of the energy cascade  process, the energy distributions  $g=f\omega$ of the 3-wave kinetic equation develop a $\delta$-like concentration at infinity.  However, since large wavenumbers correspond to highly oscillatory waves in the physical space, that cannot be described   by the kinetic theory, the full picture of the energy cascade requires a modification of the kinetic theory to be fully understood.  A similar situation can also be found in the   Boltzmann-Nordheim kinetic theory for a dilute gas of bosons \cite{Nordheim:OTK:1928}. Similar to solutions of our 3-wave kinetic equation, the solution of the Boltzmann-Nordheim kinetic equation also exhibits a singularity in finite time (cf. \cite{EscobedoVelazquez:2015:FTB,Lu:2013:TBE}). Let $\mathcal{T}$ be the time left until the solution of the Boltzmann-Nordheim equation blows up  (also in the form of a delta function). The Boltzmann-Nordheim kinetic theory applies
when $\mathcal{T}$ is still
much larger than $\hbar/\omega_0(\mathcal{T})$, where $\omega_0(\mathcal{T})$ is the average
energy of particles taking part in this blow-up.  Therefore, in the dilute gas limit,
the Boltzmann-Nordheim kinetic equation remains physically sound in the
time interval $[t_{mfp},\mathcal{T}_{cr}]$, where  $\mathcal{T}_{cr} = \hbar/\omega_0(\mathcal{T}_{cr})$ and 
 $
t_{mfp}$ is the mean-free flight time for the core of the energy
spectrum. As a consequence, the blow-up solution of the Boltzmann-Nordheim kinetic equation is  useful in describing   the BEC dynamical process but  the full dynamics of the formation and evolution of  BECs requires a serious modification of the kinetic theory (see, for instance, \cite{anglin2002bose,josserand2001nonlinear,kocharovsky2015microscopic,pomeau2000thermodynamics,PomeauBrachetMetensRica,reichl2013transport} and the references therein.)}
\item[\bf(R1.3)] Prior to this work,  several rigorous mathematical results for solutions of the 3-wave turbulence kinetic equation in various forms have already been obtained. 
Long time dynamics, hydrodynamic approximations, uniform lower bounds,  existence and uniqueness  of strong solutions and  to the quantized 3-wave turbulence kinetic equation   were discussed  in  \cite{AlonsoGambaBinh,EscobedoBinh,JinBinh,ToanBinh,ReichlTran}. The mathematical properties of solutions to the 3-wave turbulence kinetic equation under the effect of  viscosity  were studied in 
 \cite{GambaSmithBinh} for stratified  flows in the ocean and in  \cite{nguyen2017quantum} for capillary waves. A connection between chemical reaction networks and the 3-wave turbulence kinetic equation as well as its quantized version was considered  in \cite{CraciunBinh,CraciunSmithBoldyrevBinh}. We refer to the book \cite{PomeauBinh} for related discussions on the quantum versions of the 3-wave equation. 
\end{itemize}

\subsection{Description of the main results - A proof of the energy cascade phenomenon for \eqref{WeakTurbulenceInitial}-\eqref{EE1}} 
 
 In the main theorem \ref{Theorem:Existence}, we study a class of isotropic solutions $f(t,p)=f(t,|p|)$ to \eqref{WeakTurbulenceInitial}-\eqref{EE1}, in which, the initial condition is radial 
$f_0(p)=f_0(|p|)$
and 
  regular at $|p|=\infty$
$$\int_{\{|p|=\infty\}}f_0(|p|)|p|^2\omega_{|p|} d\mu(|p|) \ = \ 0,$$
and at $p=0$
$$\int_{\{|p|=0\}}f_0(|p|)|p|^2 \omega_{|p|}d\mu(|p|) \ = \ 0.$$

In this case, the energy is lost from $[0,\infty)$ and cascaded to $\{|p|=\infty\}$ and this picture can be fully described as follows. Suppose the total energy is 

$$E \ = \ \int_{[0,\infty)}f_0(|p|)\omega_{|p|}|p|^2d\mu(|p|),$$
which is a conserved quantity
$$E \ = \ \int_{[0,\infty]}f(t,|p|)\omega_{|p|}|p|^2d\mu(|p|),\ \ \ \forall t>0.$$

Then:

{
\begin{itemize}
\item[(i)]  The energy of the solution on the interval $[0,\infty)$ is a non-increasing function of time
$$\frac{d}{dt}\int_{\{|p|=\infty\}}f(t,|p|)\omega_{|p|}|p|^2d\mu(|p|) \ \ge \ 0,$$
and for all time $T_1>0$, we can always find a larger time $T_2>T_1$ such that 
$$\int_{\{|p|=\infty\}}f(T_2,|p|)\omega_{|p|} |p|^2d\mu(p) \ > \ \int_{\{|p|=\infty\}}f(T_1,|p|)\omega_{|p|}|p|^2d\mu(|p|).$$
This ensures that the energy on the interval $[0,\infty)$ keeps decreasing  for all time $t>0$. In other words, for all time $T_1>0$, we can always find a larger time $T_2>T_1$ such that 
$$\int_{[0,\infty)}f(T_2,|p|)\omega_{|p|} |p|^2d\mu(p) \ < \ \int_{[0,\infty)}f(T_1,|p|)\omega_{|p|}|p|^2d\mu(|p|).$$
It is believed  in the coagulation-fragmentation  literature that,  this kind of energy/mass loss on $[0,\infty)$ is not taken into account in the models. However, our analysis shows that the energy/mass loss is already embedded in those equations once the definition of the solutions is extended.

\item[(ii)] Moreover, the energy is transferred away from the origin as follows
$$\int_{\{0\}}f(t,|p|)\omega_p|p|^2d\mu(|p|) =0$$ for all $t\in[0,\infty)$. In addition, for all $\epsilon\in(0,1)$, there exists $R_\epsilon>0$ such that
$$\int_{\left[0, R_\epsilon \right)}f(t,|p|)\omega_{|p|}|p|^2d\mu(|p|) \ \le \ \epsilon E\text{ for all }t\in[0,\infty).$$
This inequality essentially means that the energy is cascaded away from the origin.

\item[(iii)]  The energy cascade has an explicit rate
$$\int_{\{|p|=\infty\}}f(t,|p|) \omega_{|p|}|p|^2d\mu(|p|)\ \ge \ \mathfrak{C}_1 \ - \  \frac{\mathfrak{C}_2}{\sqrt{t}},$$
where $\mathfrak{C}_1$ and $\mathfrak{C}_2$ are explicit constants. 

As a consequence, there is an explicit threshold time $T_4>0$ such that
$$\int_{\{|p|=\infty\}}f(t,|p|) \omega_{|p|}|p|^2d\mu(|p|) \ > \ 0,$$
for all time $t>T_4$. 

In other words, even when there is no energy at $\{|p|=\infty\}$ initially, after an explicit short time, some energy is lost on the interval $[0,\infty)$ and starts to accumulate at $\{|p|=\infty\}$. Therefore, one can only expect to have  local strong solutions to the equation.

\item[(iv)]  However, we can always find a time $T_5>0$ and $\mathfrak{R}_*>0$ such that for all  $t>T_5$
$$\int_{[\mathfrak{R}_*,\infty]}f(t,|p|)\omega_{|p|}|p|^2d\mu(|p|) \ \ge \ \mathfrak{C}_1.$$

\item[(v)]  At the limit $t \to \infty$, the whole energy of the system will be cascaded/lost into one single point $\{|p|=\infty\}$:
$$f(t,|p|)\omega_{|p|}|p|^2 \ \longrightarrow \ E\delta_{\{|p|=\infty\}},$$
where the limit is in the weak sense, that will be explained later in the paper.

\end{itemize}}

\subsection{An outline of the proof of the main theorem} 
 As discussed above, since we are interested in the transfer of energy to infinity, we convert the equation in $f(t,|p|)$ into an energy equation of $g(t,|p|)=f(t,|p|)\omega_{|p|}$. This is, indeed, a sophisticated coagulation-fragmentation type equation, whose fragmentation term is nonlinear rather than linear  (see Definition \ref{Def:WeakSolution} and (R1.1)).  Section 5.1 is devoted to the proof of the existence of weak solutions of the equation of $g$, based on a regularization technique of the collision operator. The proof of the energy cascade phenomenon relies on a very special structure of the collision operator of $g$. We will present in Section 4 that for test functions of the form $\phi_r(p)=\left(1-\frac{r}{p}\right)_+$, with any positive parameter $r>0$, the collision operator becomes positive. The key property of these test functions is that $\lim_{p\to\infty}\phi_r(p)=1$. Using this class of test functions, we can see that $\partial_t\int_{[0,\infty]}g(t,|p|)\phi_r(p)d\mu(|p|)\ge 0$ for all $r>0$, that yields
$$\inf_{r>0}\partial_t\int_{[0,\infty]}g(t,|p|)\phi_r(p)d\mu(|p|)=\partial_t\int_{\{\infty\}}g(t,|p|)\phi_r(p)d\mu(|p|)\ge 0.$$ 
This is an indicator for the accumulation of the energy at infinity and the main idea of Section 5.2.  Based on this special structure of the collision operator, in Section 5.4, we compare the energy measured on an interval $[r,\infty]$ at time $t$  and the energy measured on the interval $[0,r]$ for all time $s$  from $0$ to $t$. This leads to the proof of the transfer of energy away from $\{0\}$ in Section 5.5. By comparing the energy from $[0,\infty)$ and the energy from $[0,\infty]$, using the observations from Section 5.2 and Section 5.4, we can prove in Section 5.6 that the some energy indeed is lost from the interval $[0,\infty)$ in finite time. Sections 5.7 and 5.8 study the lost rate and show that all of the energy will be finally accumulated at $\{\infty\}$ as time evolves.

\section{New weak formulation and statement of the main results}
\subsection{New weak formulation, energy conservation and H-Theorem on the extended  space}\label{Sec:Preliminaries}
\subsubsection{The new weak formulation on the extended space}
Since the energy is  finite, and we are interested in the cascade of this conserved quantity from low to high wavenumbers as time evolves, we denote the energy distribution by $g=f\omega_p=f|p|$,  and obtain a new form of  the collision operator
\begin{align}\label{EE2}
\begin{split}
\tilde{Q}\left[{g}\right]:\ =\ Q\left[\frac{g}{p}\right]\ =\ & \ \iint _{ \mathbb{R}^3\times\mathbb{R}^3} |V_{p,p_1,p_2}|^2\delta (   |p|-|p_1|-|p_2|    )  \delta (p-p_1-p_2)\times
\\
&\times\left[ \frac{g_1}{|p_1|}\frac{g_2}{|p_2|}-\left(\frac{g_1}{|p_1|}+\frac{g_2}{|p_2|}\right)\frac{g}{|p|}\right]{d}^3p_1{d}^3p_2 \\
& - \ 2\iint _{ \mathbb{R}^3\times\mathbb{R}^3} |V_{p,p_1,p_2}|^2\delta (   |p_1|-|p|-|p_2|    )  \delta (p_1-p-p_2)\times
\\
&\times\left[ \frac{g}{|p|}\frac{g_2}{|p_2|}-\left(\frac{g}{|p|}+\frac{g_2}{|p_2|}\right)\frac{g_1}{|p_1|}\right]{d}^3p_1{d}^3p_2 
\end{split}
\end{align}
The {\it energy evolution equation} for the energy distribution $g$ is now

\begin{equation}\label{WeakTurbulenceEnergy}
\begin{aligned}
\partial_tg(t,p) \ =& |p|\tilde{Q}\left[{g}\right] \ = \ |p|Q\left[\frac{g}{p}\right](t,p), \\\
g(0,p) \ =& \ g_0(p),
\end{aligned}
\end{equation}
where $g_0(p)=|p|f_0(p)$.

 In Definition \ref{Def:WeakSolution}, we will {\it define the weak solution $f$ to \eqref{WeakTurbulenceInitial} in terms of the weak solution $g$} to the energy evolution equation \eqref{WeakTurbulenceEnergy}. Before presenting the form of the new weak formulation, we recall some of the definitions of function  and extended measure spaces.

\begin{definition}[Continuous Function Spaces]\label{Def:ContinuousSpaces}~~

\begin{itemize}
\item Let $I$ be one of the intervals  $[a,\infty]$, $(a,\infty]$, $[a,\infty)$ or $(a,\infty)$ with $0\leq a<\infty$, we denote by $C(I)$ the set of functions that are continuous on $I$; by $C^k(I)$, with $k\in\mathbb{N}\cup\{0\}$, the set of functions in $C(I)$ for which the derivatives of order up to $k$ exist and are in $C(I)$; and by $C_c^k(I)$ the set of functions in $C^k(I)$ supported in a compact $K\subset I$. 
\item We  denote $C(I)=C^0(I)$ and $C_c(I)=C_c^0(I)$.
\item We define the space $\mathfrak{M}$ to be the function space spanned  the space
$$\Big\{ \phi(p) \ \ \Big| \ \ p\phi\in C_c([0,\infty)) \Big\}.$$
and the space $C([0,\infty])$, in which for each $\phi\in C([0,\infty])$ the limit  $\lim_{p\to\infty}\phi(p)$ exists.
\item
For $\psi\in C(I)$, we define  $$\|\psi\|_{L^\infty}=\|\psi\|_{C(I)}.$$
\end{itemize}
\end{definition}
\begin{definition}[Extended Measure]\label{Def:RadonMeasure}~~

\begin{itemize}
\item The ``extended half real line'' is the set $[0,\infty] \ = \ [0,\infty)\cup\{ \infty\}$
with the topology generated by the open sets of $\mathbb{R}$ and all interval $[0,a)$ and $(a,\infty]$. Then denoted by $\mathfrak{B}([0,\infty))$ the set of Borel sets of $[0,\infty)$, we can define the set of Borel sets of the extended real line by
$$\mathfrak{B}([0,\infty]) \ =  \ \{B\subset [0,\infty]: B\cap [0,\infty) \in \mathfrak{B}([0,\infty))\}.$$   
\item Note that $[0,\infty]$ is a Hausdorff space. By $\mathfrak{D}([0,\infty])$, we denote the space of finite nonnegative  measures  on $\mathfrak{B}([0,\infty])$. This classical construction  could be found in textbooks in analysis, for instance \cite{folland1995introduction}.
\item For any interval $I$ of the form $[a,\infty]$, $(a,\infty]$, $[a,\infty)$ or $(a,\infty)$ with $0\leq a<\infty$, we denote by $\mathfrak{D}(I)$  the space of finite nonnegative  measures $\mu\in\mathfrak{D}([0,\infty])$ such that $\mu\equiv0$ on $[0,\infty]\setminus I$.

\item In the notation of integrals, we  write $\varrho(x)d\mu( x)$ and  for any $\varrho\in\mathfrak{D}([0,\infty])$, $$\|\varrho\|_{L^1}=\int_{[0,\infty]}|\varrho(x)|d\mu( x).$$
\item The delta function $\delta_{\{x=\infty\}}$ satisfies
\begin{equation}\label{DeltaInfty1}
\int_{[0,\infty]}\delta_{\{x=\infty\}} \psi(x)d\mu(x) \ = \ \lim_{x\to\infty}\psi(x), 
\end{equation} 
for all continuous function $\psi\in C([0,\infty))$ such that the above limit exists.

If furthermore $\psi\in C([0,\infty])$, 
\begin{equation}\label{DeltaInfty2}
\int_{[0,\infty]}\delta_{\{x=\infty\}} \psi(x)d\mu(x) \ = \ \psi(\infty). 
\end{equation}

\end{itemize}

\end{definition}

\begin{definition}[Weak$^*$ Topologies]\label{Def:WeakStar}~~~
\begin{itemize}

\item We define the weak$^*$ topology on $\mathfrak{D}(I)$ to be the smallest topology such that the mapping $$\nu\in\mathfrak{D}(I)\mapsto\int_I\psi(x)\nu(x)d \mu(x)$$ is continuous for all test functions $$\psi\in C_0(I):=\{\psi\in C(\bar{I}):\psi\equiv0\text{ on }\bar{I}\setminus I\}.$$

\item It can be shown  that the space $C_0(I)$,  endowed with the supremum norm, is a separable Banach space. Therefore, by the Banach-Alaoglu theorem, the unit ball in $\mathfrak{D}(I)$ is compact with respect to the weak$^*$ topology and the weak$^*$~topology is metrizable. Hence, we can define:

A sequence $\{\nu_n\}$ in $\mathfrak{D}(I)$ is said to converge to $\nu$ with respect to the weak$^*$ topology
$$\nu_n \stackrel{\ast}{\rightharpoonup} \nu$$
 if and only if $$\int_{I}\psi(x)\nu_n(x)d \mu(x)\rightarrow\int_{I}\psi(x)\nu(x)d \mu(x)$$ for all $\psi\in C_0(I)$.

\item We endow $\mathfrak{D}(I)$ with the weak$^*$ topology.

\end{itemize}

\end{definition}

\begin{remark}
The reason that we need to extend the interval $[0,\infty)$ to $[0,\infty]$ is that weak solutions to \eqref{WeakTurbulenceInitial}-\eqref{EE1}  defined on $[0,\infty)$, whose energy is conserved on $[0,\infty)$, are only local in time. This can be seen from the proof of Proposition \ref{Propo:Mass0}.  To guarantee the existence of global in time solutions, one needs a weaker definition of the concept of weak solutions. In this case, the weak solutions need to be considered on the extended real line  $[0,\infty]$. 
\end{remark}
Next, we represent the weak formulation on the real line. Based on this formulation, the weak solutions on the extended real line will be introduced. 
\begin{proposition}[Weak Formulation  on the Real Line]\label{Lemma:WeakFormulation}
 { For any suitable test function $\phi(p)$, the following weak  formulation holds true for the collision operator \eqref{EE1} }
\begin{align}\label{Lemma:WeakFormulation:Eq1}
\begin{split}
\int_{\mathbb{R}^3}&  Q[f]|p|\phi d^3p\,=\  \int_{\mathbb{R}^3}Q\left[\frac{g(p)}{|p|}\right]|p|\phi d^3p\,= \int_{\mathbb{R}^3}\int_{\mathbb{R}^3}\int_{\mathbb{R}^3} {|p| |p_1| |p_2|}\delta(p-p_1-p_2)\\
&\times \delta(|p|-|p_1|-|p_2|)\left[ \frac{g(p_1)}{|p_1|}\frac{g(|p_2|)}{|p_2|}-\frac{g(p_1)}{|p_1|}\frac{g(p)}{|p|}-\frac{g(p_2)}{|p_2|}\frac{g(p)}{|p|} \right]\\
&\hspace{1cm}\times \Big[|p|\phi(p)-|p_1|\phi(p_1)-|p_2|\phi(p_2)\Big]d^3p\,d^3p_1\,d^3p_2\,\\
&\hspace{-.5cm}=2\pi\int_{\mathbb{R}^3}\int_{\mathbb{R}^+}\,{\big|p_1+|p_2|\widehat{p_1}\big| |p_1| |p_2|^{3}}\Big[ \frac{g(p_1)}{|p_1|}\frac{g(|p_2|\widehat{p_1})}{|p_2|} \\
& -\frac{g(p_1)}{|p_1|}\frac{g(p_1+|p_2|\widehat{p_1})}{|p_1+|p_2|\widehat{p_1}|}-\frac{g(|p_2|\widehat{p_1})}{||p_2|\widehat{p_1}|}\frac{g(p_1+|p_2|\widehat{p_1})}{|p_1+|p_2|\widehat{p_1}|}\Big]\\
&\times \Big[|p_1+|p_2|\widehat{p_1}|\phi(p_1+|p_2|\widehat{p_1})-|p_1|\phi(p_1)-|p_1|\phi(|p_2|\widehat{p_1})\Big]d^3p_1d|p_2|,
\end{split}
\end{align}
in which $\hat{p}=\frac{p}{|p|}$.

In addition, for radially symmetric functions $f(p):=f(|p|)$, $g(p):=g(|p|)$ and $\phi(p):=\phi(|p|)$, the following holds true
\begin{align}\label{WeakFormulation:radialL}
\begin{split}
& \  \int_{\mathbb{R}^3}\,Q\left[\frac{g(p)}{|p|}\right]|p|\phi  d^3p   =  \ 16\pi^2\int_{|p_1|>|p_2|\ge0}{|p_1||p_2|}g(|p_1|)g(|p_2|)\times \\
 & \  \  \  \  \  \  \  \  \  \  \  \  \times \Big[|p_1+p_2|^3\phi(|p_1|+|p_2|)-2(p_1^2+p_2^2)|p_1|\phi(|p_1|)\\
 & \  \  \  \  \  \  \  \  \  \  \  \   -4p_1p_2^2\phi(|p_2|)+(|p_1|-|p_2|)^3\phi(|p_1|-|p_2|)\Big]d|p_{1}|\,d|p_2|\,\\
 & \  \  \  \  \  \  \  \  \  \  \  \   + 8\pi^2\int_{|p_1|=|p_2|\ge0}{|p_1||p_2|}g(|p_1|)g(|p_2|)(|p_1|+|p_2|)^2\Big[(|p_1|+|p_2|)\phi(|p_1|+|p_2|)\\
 &\  \  \  \  \  \  \  \  \  \  \  \  -|p_1|\phi(|p_1|)-|p_2|\phi(|p_2|)\Big]d|p_{1}|\,d|p_2|\,.
\end{split}
\end{align}

In the rest of the paper, for the sake of simplicity, we omit the factor $8\pi^2$.
\end{proposition}
%
%
\begin{proof} The proof follows the  arguments of \cite{AlonsoGambaBinh,nguyen2017quantum}.
\end{proof}

\begin{definition}[The functionals $\mathfrak{H}^1_\phi$ and $\mathfrak{H}^2_\phi$]\label{Def:L1L2}

For a function $\phi\in  \mathfrak{M}$,  and $x\ge y\ge 0$, we define
\begin{equation}\label{L1}
\mathfrak{H}^1_\phi(x,y) \ = \ |x+y|^3\phi(x+y)-2(x^2+y^2)x\phi(x)-4xy^2\phi(y)+|x-y|^3\phi(x-y),
\end{equation}
and
\begin{equation}\label{L2}
\mathfrak{H}^2_\phi(x,y) \ = \ |x+y|^2[(x+y)\phi(x+y)-x\phi(y)-y\phi(y)].\end{equation}
Then  formula \eqref{WeakFormulation:radialL} can be written as
\begin{align}\label{WeakFormulation:radialLNew}
\begin{split}
  \int_{\mathbb{R}^3}\,Q\left[\frac{g}{|p|}\right]|p|\phi d^3p
 =& \ 2\int_{|p_1|>|p_2|\ge0}{|p_1||p_2|}g(|p_1|)g(|p_2|)\mathfrak{H}^1_\phi(|p_1|,|p_2|){d}|p_{1}|\,{d}|p_2|\,\\
 & \ + \int_{|p_1|=|p_2|\ge0}{|p_1||p_2|}g(p_1)g(p_2)\mathfrak{H}^2_\phi(|p_1|,|p_2|) {d}|p_{1}|\,{d}|p_2|\,.
\end{split}
\end{align}
\end{definition}

\begin{definition}[Weak Solution in terms of Energy Distribution on the Extended Real Line] \label{Def:WeakSolution}

From now on, with an abuse of notations, we use $p,p_1,p_2$ for  real and positive numbers. 

\begin{itemize}
\item[(i)] Suppose that 
$$\int_{[0,\infty]}g_0(p)p^2 d\mu(p)<\infty.$$
A function $g(t,p)$, such that $g(t,p)|p|^2\in C([0,\infty):\mathfrak{D}([0,\infty]))$ for all $t\in[0,\infty)$ and for all $\phi\in C^1([0,\infty):C([0,\infty]))$, $g$ satisfies
\begin{equation}\label{Def:WeakSolution:E1}
\begin{split}
&\int_{[0,\infty]}\phi(t,p)g(t,p)|p|^2d \mu(p)-\int_{[0,\infty]}\phi(0,p)g_0(p)|p|^2d \mu(p)\\
&\indent\begin{split}=\int_0^t\bigg[&\int_{[0,\infty]}\phi_s(s,p)g(s,p)|p|^2d \mu(p)\\&\indent+2\int_{p_1>p_2\ge 0}{g(s,p_1)|p_1|g(s,p_2)|p_2|}\mathfrak{H}^1_{\phi(s,\cdot)}(p_1,p_2)d \mu(p_1)d \mu(p_2) + \\
& \indent+ \int_{p_1=p_2\ge 0}{g(s,p_1)|p_1|g(s,p_2)|p_2|}\mathfrak{H}^2_{\phi(s,\cdot)}(p_1,p_2)d \mu(p_1)d \mu(p_2)  \bigg]d s\end{split}
\end{split}
\end{equation}
will be called a {\bf weak solution} to \eqref{EE2}-\eqref{WeakTurbulenceEnergy}. Then 
$$
\begin{aligned}
f(t,p)\  = \  & \frac{g(t,p)}{|p|}, \mbox{ for } p>0, t\ge 0,\\
f(t,0)\  = \  & f_0(0), \mbox{ for } t\ge 0.
\end{aligned}$$ 
will be called  a {\bf weak solution} to \eqref{WeakTurbulenceInitial}-\eqref{EE1}.
\item[(ii)] Suppose that 
$$\int_{[0,\infty]}g_0(p)p^2 d\mu(p)<\infty,$$
and
$$\int_{[0,\infty]}g_0(p)p d\mu(p)<\infty.$$
A function $g(t,p)$, such that $g(t,p)|p|^2\in C([0,\infty):\mathfrak{D}([0,\infty]))$ for all $t\in[0,\infty)$ and for all $\phi\in C^1([0,\infty):\mathfrak{M})$, $g$ satisfies
\begin{equation}\label{Def:WeakSolution:E1}
\begin{split}
&\int_{[0,\infty]}\phi(t,p)g(t,p)|p|^2d \mu(p)-\int_{[0,\infty]}\phi(0,p)g_0(p)|p|^2d \mu(p)\\
&\indent\begin{split}=\int_0^t\bigg[&\int_{[0,\infty]}\phi_s(s,p)g(s,p)|p|^2d \mu(p)\\&\indent+2\int_{p_1>p_2\ge 0}{g(s,p_1)|p_1|g(s,p_2)|p_2|}\mathfrak{H}^1_{\phi(s,\cdot)}(p_1,p_2)d \mu(p_1)d \mu(p_2) + \\
& \indent+ \int_{p_1=p_2\ge 0}{g(s,p_1)|p_1|g(s,p_2)|p_2|}\mathfrak{H}^2_{\phi(s,\cdot)}(p_1,p_2)d \mu(p_1)d \mu(p_2)  \bigg]d s\end{split}
\end{split}
\end{equation}
will be called a {\bf weak solution} to \eqref{EE2}-\eqref{WeakTurbulenceEnergy}. Then 
$$
\begin{aligned}
f(t,p)\  = \  & \frac{g(t,p)}{|p|}, \mbox{ for } p>0, t\ge 0,\\
f(t,0)\  = \  & f_0(0), \mbox{ for } t\ge 0.
\end{aligned}$$ 
will be called  a {\bf weak solution} to \eqref{WeakTurbulenceInitial}-\eqref{EE1}.
\item[(iii)] The above two definitions differ in the spaces $\mathfrak{M}$ and $C([0,\infty])$.
\end{itemize}

\end{definition}
\subsubsection{Energy conservation and H-Theorem on the extended measure space}
In the proposition below, we show that there is a conservation of energy for the weak solutions  in the sense of Definition \ref{Def:WeakSolution}. This conservation of energy is on $[0,\infty]$. In other words, following Proposition \ref{Propo:Mass0}, there is a loss of energy from $[0,\infty)$ to one single point $\{\infty\}$ in the mathematical framework considered in our paper. 

From the proof of the main  Theorem \ref{Theorem:Existence}, we can see that weak solutions, defined in the classical sense on $[0,\infty)$, are only local in time. In other words, global in time solutions, whose energy is conserved inside the interval $ [0,\infty)$, in general, do not exist, even in the classical weak sense.
\begin{proposition}\label{Propo:HTheorem}
Let $f$ be a weak solution to \eqref{WeakTurbulenceInitial}-\eqref{EE1}  in the sense of Definition \ref{Def:WeakSolution}. Then the following conservation of energy holds true
\begin{equation}\label{Propo:HTheorem:1}
\int_{[0,\infty]}p^3f(t,p)d\mu(p) \ = \ \int_{[0,\infty]}p^3f_0(p)d\mu(p).
\end{equation}
Moreover, the H-theorem also formally holds on $[0,\infty]$
\begin{equation}\label{Propo:HTheorem:2}
\partial_t\int_{[0,\infty]}p^2\log[f(t,p)]d\mu(p) \ \le \ 0.
\end{equation}
\end{proposition}
\begin{proof}
This can be proved using the  argument of \cite{nguyen2017quantum,germain2017optimal}.

\end{proof}
\subsection{Main results}

\begin{theorem}[Global Existence and Energy Cascade]\label{Theorem:Existence}
Given any $f_0p^3\in\mathfrak{D}([0,\infty])$, $f_0\ge 0$ satisfying
$\int_{[0,\infty]}p^3f_0(p)d \mu(p)<\infty,$
 there exists at least one weak solution $fp^3\in C([0,\infty):\mathfrak{D}([0,\infty]))$ in the sense of Definition \ref{Def:WeakSolution} (i), $f\ge0$
 to \eqref{WeakTurbulenceInitial}-\eqref{EE1} that satisfies $f(0,\cdot)=f_0$. Moreover,
\begin{equation}\label{Theorem:Existence:1}
\int_{[0,\infty]}p^3f(t,p)d \mu(p)=\int_{[0,\infty]}p^3f(0,p)d \mu(p)\text{ for all }t\in[0,\infty).
\end{equation}
If, in addition $\int_{[0,\infty]}p^2f_0(p)d\mu( p)<\infty,$
 there exists at least one weak solution $fp^3\in C([0,\infty):\mathfrak{D}([0,\infty]))$ in the sense of Definition \ref{Def:WeakSolution} (ii), $f\ge0$
 to \eqref{WeakTurbulenceInitial}-\eqref{EE1} that satisfies $f(0,\cdot)=f_0$. These solutions have conserved energy and bounded mass on $[0,\infty]$ for all time $t>0$,
\begin{equation}\label{Theorem:Existence:1a}
\int_{[0,\infty]}p^2f(t,p)d \mu(p)\le \int_{[0,\infty]}p^2f(0,p)d \mu(p)\text{ for all }t\in[0,\infty).
\end{equation}
\begin{equation}\label{Theorem:Existence:1b}
\int_{[0,\infty]}p^3f(t,p)d \mu(p)=\int_{[0,\infty]}p^3f(0,p)d \mu(p)\text{ for all }t\in[0,\infty).
\end{equation}
Moreover, if 
$$\int_{\{0\}}p^3f_0(p)d \mu(p)=0, \mbox{ and }
\int_{\{\infty\}}p^3f_0(p)d \mu(p)=0,$$
the followings hold true.
\begin{itemize}
\item[(i)]  Given any non-trivial weak solution $f$  to  with initial condition $f_0$ in the sense of Definition \ref{Def:WeakSolution}, $f\ge0$, then  the mapping $t\mapsto\int_{\{\infty\}}f(t,p)p^3d \mu(p)$ is nondecreasing.

\item[(ii)] Given any non-trivial weak solution $f$  to  with initial condition $f_0$ in the sense of Definition \ref{Def:WeakSolution},  $f\ge0$. For any ${t}_0\in[0,\infty)$, there exists $t_1>t_0$ such that
\begin{equation}\label{Theorem:Existence:3}
\int_{\{\infty\}}f(t_1,p)|p|^3d \mu(p)>\int_{\{\infty\}}f({t}_0,p)|p|^3d \mu(p),
\end{equation}
which means there exists $t_*>0$ such that $\int_{\{\infty\}}f(t_*,p)|p|^3d \mu(p)>0,$
and
$\int_{\{0\}}f(t,p)|p|^3d \mu(p)=0$ for all $t\in[0,\infty)$. Moreover, for all $\epsilon\in(0,1)$, there exists $R_\epsilon>0$ satisfying
$$\int_{\left[0, R_\epsilon \right)}f(t,p)p^3d \mu(p) \ \le \ \epsilon  \|f(0,p)|p|^3\|_{L^1}\text{ for all }t\in[0,\infty).$$
\item[(iii)]  Given any non-trivial weak solution $f$  with initial condition $f_0$ in the sense of Definition \ref{Def:WeakSolution},  $f\ge0$. There exist explicit constants $C_1$, $C_2$, $T^*>0$ depending on the initial condition $f_0$ such that   the following inequality holds
\begin{equation*}
\int_{\{\infty\}}f(t,p)|p|^3d \mu(p)\ >\ C_1 \ - \ \frac{C_2}{\sqrt{t}},
\end{equation*}
for all $t>T^*$. 

That leads to the existence of an explicit $T^{**}>0$ such that 
\begin{equation*}
\int_{\{\infty\}}f(t,p)|p|^3d \mu(p)\ >\ \frac{C_1}{2},
\end{equation*}
for all $t>T^{**}$. 

Moreover, there exist  $r>0$ and   an explicit $T_r>0$ depending on $r$ such that  for all $t>T_r$ 
\begin{equation*}
\int_{[r,\infty]}f(t,p)|p|^3d \mu(p)\ \ge \ C_1.
\end{equation*}
\item[(iv)]   All solutions $f$ with initial condition $f_0$ in the sense of Definition \ref{Def:WeakSolution},  $f\ge0$, with finite energy, converge weakly$^*$ in $\mathfrak{D}([0,\infty])$ to a Dirac measure at infinity as $t\rightarrow\infty$ i.e. $$f(t,p)|p|^3\stackrel{\ast}{\rightharpoonup} \|f(0,p)|p|^3\|_{L^1}\delta_{\{p=\infty\}}$$ as $t\rightarrow\infty$.
\item[(v)] Suppose that $f$ is a  weak solution   in the sense of Definition \ref{Def:WeakSolution}, satisfying $f(0,p)p^3\equiv E\delta_{\{p=\infty\}}$ for some $E\in[0,\infty)$. Then $f$ is a trivial solution in the sense $f(t,p)p^3\equiv E\delta_{\{p=\infty\}}$ for all $t\in[0,\infty)$.
\end{itemize}

\end{theorem}

%

\begin{corollary}\label{Col}
If
$$\int_{[0,\infty]}p^3f_0(p)d \mu(p)<\infty,$$
and
$$\int_{[0,\infty]}p^2f_0(p)d \mu(p)<\infty,$$
then the weak solution found in Theorem \ref{Theorem:Existence} is also the weak solution in the classical sense i.e. $f$ satisfies 
\begin{equation}\label{Def:WeakSolution:E1}
\begin{aligned}
&\int_{[0,\infty]}\psi(t,p)f(t,p)|p|^2d \mu(p)-\int_{[0,\infty]}\psi(0,p)f_0(p)|p|^2d \mu(p)\\
= &\int_0^t\bigg[\int_{[0,\infty]}\psi_s(s,p)f(s,p)|p|^2d \mu(p)\\&\indent+2\int_{p_1>p_2\ge 0}{f(s,p_1)|p_1|^2f(s,p_2)|p_2|^2}\Big[|p_1+p_2|^2\psi(p_1+p_2)-2(p_1^2+p_2^2)\psi(p_1)\\
&-4p_1p_2\psi(p_2)+|p_1-p_2|^2\psi(p_1-p_2)\Big]d \mu(p_1)d \mu(p_2) + \\
&+ \int_{p_1=p_2\ge 0}{f(s,p_1)|p_1|^2f(s,p_2)|p_2|^2}|p_1+p_2|^2\big[\psi(p_1+p_2)-\psi(p_1)-\psi(p_2)\big]d \mu(p_1)d \mu(p_2)  \bigg]d s,
\end{aligned}
\end{equation}
for all $\psi\in C_c([0,\infty))$.
\end{corollary}
\begin{proof}
The proof follows from straightforward computations and the fact that if $\psi\in C_c([0,\infty))$ then $\psi/p\in \mathfrak{M}$. 
\end{proof}

\section{Properties of the collision operator and a special class of test functions}\label{Sec:CollisionStructure}
This section is devoted to the construction of test function $\phi$ such that $\mathfrak{H}^1_\phi$ and $\mathfrak{H}^2_\phi$ are positive. These test functions play a crucial role in proving the cascade of energy to infinity. 
\subsection{Boundedness of the functionals $\mathfrak{H}^1_\phi$ and $\mathfrak{H}^2_\phi$}
In the propositions below, we bound $\mathfrak{H}^1_\phi$ and $\mathfrak{H}^2_\phi$ in terms of the norm of the test functions $\phi$. These estimates will be used later in proving the existence of weak solutions in the sense of Definition \ref{Def:WeakSolution}.
\begin{proposition}\label{Propo:L1Boundedness}
For a function  $\phi\in C([0,\infty])$ and $p^2\phi'(p)\in L^\infty([0,\infty]),$ $\phi'(p)$ is piece-wise continuous on $[0,\infty)$ and $\exists \mathfrak{C}>0$ such that $$p_2|\phi(p_1+p_2)-\phi(p_1)|\le\mathfrak{C}$$ for $p_1\ge p_2\ge 0$.
 Denote the set of all $\phi$ by $\mathcal{V}([0,\infty])$, then $\mathcal{V}$
is a vector space with the following norm
$$\|\phi\|_{\mathcal{V}} \ = \ \left\|\phi\right\|_{L^\infty} \ + \left\|p^2\phi'(p)\right\|_{L^\infty} \ + \ \sup_{p_1\ge p_2\ge 0}p_2|\phi(p_1+p_2)-\phi(p_1)|$$ then
 $$|\mathfrak{H}^1_\phi(p_1,p_2)| \ \le \ 10 p_1p_2\|\phi\|_{\mathcal{V}}.$$
 
\end{proposition}
\begin{proof}
First, rewrite the form of $\mathfrak{H}^1_\phi$ for $p_1\ge p_2\ge 0$
\begin{equation}
\label{Propo:L1Boundedness:E1}
\begin{aligned}
\mathfrak{H}^1_\phi(p_1,p_2) \ = & \ (p_1+p_2)^3 \phi(p_1+p_2) \ - \ (2p_1^3+2p_1p_2^2)\phi(p_1)\\
& \ - 4 p_1p_2^2\phi(p_2) \ + \ (p_1-p_2)^3\phi(p_1-p_2).
\end{aligned}
\end{equation}

We add and subtract at the same time the above identity by $(p_1-p_2)^3\phi(p_1+p_2)$
\begin{equation*}
\begin{aligned}
\mathfrak{H}^1_\phi(p_1,p_2) \ = & \ [(p_1+p_2)^3+(p_1-p_2)^3] \phi(p_1+p_2) \ - \ (2p_1^3+2p_1p_2^2)\phi(p_1)\\
& \ - 4 p_1p_2^2\phi(p_2) \ + \ (p_1-p_2)^3[\phi(p_1-p_2)-\phi(p_1+p_2)],
\end{aligned}
\end{equation*}
and  estimate the absolute value of the last term on the right hand side of the new identity
\begin{equation*}
\begin{aligned}
|(p_1-p_2)^3[\phi(p_1-p_2)-\phi(p_1+p_2)]| 
\ = & \ |p_1-p_2|^3\left|\int_{p_1-p_2}^{p_1+p_2}\frac{\xi^2\phi'(\xi)}{\xi^2}d\xi\right|,
\end{aligned}
\end{equation*}
in which the integral $\int_{p_1-p_2}^{p_1+p_2}$ is defined in the following sense: Suppose that $\phi'(\xi)$ is discontinuous at the points $a_1,\cdots,a_l$ in the interval $[p_1-p_2,p_1+p_2]$, then
\begin{equation}
\label{Discon}
\int_{p_1-p_2}^{p_1+p_2} \ = \ \int_{p_1-p_2}^{a_1} \ + \ \cdots \ + \int_{a_l}^{p_1+p_2}.
\end{equation}
Since in the above integral, the values of $\xi$ is taken in the interval $[p_1-p_2,p_1+p_2]$, it is straightforward that $\frac{1}{\xi^2}\le \frac{1}{|p_1-p_2|^2},$
which implies
\begin{equation*}
\begin{aligned}
|(p_1-p_2)^3[\phi(p_1-p_2)-\phi(p_1+p_2)]| \ \le & \ |p_1-p_2|^3\left|\int_{p_1-p_2}^{p_1+p_2}\phi'(\xi)d\xi\right|
\ \le  \ 2p_1p_2\left\|\xi^2\phi'(\xi)\right\|_{L^\infty},
\end{aligned}
\end{equation*}
where the integral is defined in the sense of \eqref{Discon}.

Combining the last two inequalities, we find the following bound on $\mathfrak{H}^1_\psi(p_1,p_2)$
 \begin{equation*}
\begin{aligned}
\Big|\mathfrak{H}^1_\phi(p_1,p_2)\Big| 
\ \le & \ \Big|[2p_1^3+2p_1p_2^2+4p_1p_2^2] \phi(p_1+p_2) \ - \ (2p_1^3+2p_1p_2^2)\phi(p_1)\\
& \ - 4 p_1p_2^2\phi(p_2)\Big| \ + \ 2p_1p_2\left\|\xi^2\phi'(\xi)\right\|_{L^\infty},
\end{aligned}
\end{equation*}
which can be rewritten under the form
\begin{equation}\label{Propo:L1Boundedness:E2}
\begin{aligned}
\Big|\mathfrak{H}^1_\phi(p_1,p_2) \Big|
\ \le & \ \mathcal{A} \ + \  \mathcal{B} \ + \ 2p_1p_2\left\|\xi^2\phi'(\xi)\right\|_{L^\infty},
\end{aligned}
\end{equation}
where
$$\mathcal{A} \ := \ \Big|(2p_1^3+2p_1p_2^2)[\phi(p_1+p_2)-\phi(p_1)]\Big| \mbox{
and }
\mathcal{B} \ := \ \Big|4 p_1p_2^2[\phi(p_1+p_2)-\phi(p_2)]\Big|.$$

Let us now estimate $\mathcal{A} $
\begin{equation*}
\mathcal{A} \ =  \ \Big|(2p_1^3+2p_1p_2^2)\int_{p_1}^{p_1+p_2}\phi'(\xi)d\xi\Big|\
\ =  \ \Big|(2p_1^3+2p_1p_2^2)\int_{p_1}^{p_1+p_2}\frac{\xi^2\phi'(\xi)}{\xi^2}d\xi\Big|,
\end{equation*}
where the integral is defined in the sense of \eqref{Discon}.

Observe that $\xi$ in the integral is taken within the interval $[p_1,p_1+p_2]$, then
$\frac{1}{\xi^2}\le \frac{1}{p_1^2},$
that means
\begin{equation}\label{Propo:L1Boundedness:E3}
\begin{aligned}
\mathcal{A} \ \le & \  \frac{2p_1^2+2p_2^2}{p_1}\int_{p_1}^{p_1+p_2}{\xi^2\phi'(\xi)}d\xi
\ \le  \  4p_1p_2\left\|\xi^2\phi'(\xi)\right\|_{L^\infty},
\end{aligned}
\end{equation}
where the integral is defined in the sense of \eqref{Discon}.

It is straightforward to bound $\mathcal{B}$
\begin{equation}\label{Propo:L1Boundedness:E4}
\begin{aligned}
\mathcal{B} \ \le & \ 4 p_1p_2\mathfrak{C}.
\end{aligned}
\end{equation}

Combining \eqref{Propo:L1Boundedness:E2}, \eqref{Propo:L1Boundedness:E3} and \eqref{Propo:L1Boundedness:E4} gives the final conclusion of the Proposition.
\end{proof}
\begin{proposition}\label{Propo:L2Boundedness}

 For a function  $\phi\in\mathcal{V}([0,\infty))$, where $\mathcal{V}$ is defined in Proposition \ref{Propo:L1Boundedness}, then
 $$|\mathfrak{H}^2_\phi(p,p)| \ \le \ 8p^2\left\|\phi\right\|_{\mathcal{V}}.$$
\end{proposition}
\begin{proof}
First, rewrite the form of $\mathfrak{H}^2_\phi(p,p)$  in terms of $\phi$
\begin{equation*}
\mathfrak{H}^2_\phi(p,p) \ = \ 8p^2[\phi(2p)-\phi(p)].
\end{equation*}

The same argument used in the previous proposition gives
\begin{equation*}
|\mathfrak{H}^2_\phi(p,p)| \ = \ 8p^3\int_{p}^{2p}\phi'(\xi)d\xi \ \le \ 8p^2\left\|\phi\right\|_{\mathcal{V}}.
\end{equation*}
\end{proof}

\begin{proposition}\label{Propo:L3Boundedness}
Define the space $\mathcal{W}([0,\infty))$ to be the vector space spanned by $$\mathcal{V}_0([0,\infty)) \ = \ \Big\{ \psi(p)\ = \ p\phi(p) \ \ \Big| \ \phi\in \mathcal{V}([0,\infty))\Big\} \mbox{ and the vector } \psi \ = \ 1.$$ We also define $$\mathcal{W}^0([0,\infty)) \ = \ \Big\{\phi \ \Big|\  p\phi\in \mathcal{W}([0,\infty))  \Big\}.$$

Then $\mathcal{V}([0,\infty))$ is dense in $C([0,\infty])$ with respect to the $L^\infty$-norm. And $\mathcal{W}([0,\infty))$ is also dense in  $C_c([0,\infty))$ with respect to the $L^\infty$-norm.
\end{proposition}
\begin{proof}

Let us define $S$ to be the vector space spanned by $C^1_c((0,\infty))$ and the two functions $\frac{1}{p+1}$ and $1$. Since   $S \subset\mathcal{V}([0,\infty))$ and $S$ is dense in $C([0,\infty])$, we deduce that $\mathcal{V}([0,\infty))$ is dense in $C([0,\infty])$.

Now, let us consider a function $\psi\in C_c([0,\infty))$. Observe that $\psi$ can be decomposed as the sum  of $\psi_1\in C([0,a])$ for some $0<a<\infty$ and $\psi_2\in C_c(0,\infty)$. It is clear that $\psi_2(p)$ can be approximated as the limit of a sequence $\{\psi_{2,n}\}$ in $C_c^1(0,\infty)$. Since $\{\psi_{2,n}/p\}\subset C_c^1(0,\infty)\subset \mathcal{V}([0,\infty))$. As a consequence $\psi_2$ can be approximated in the $L^\infty$ norm by a sequence in $\mathcal{V}_0([0,\infty))\subset \mathcal{W}([0,\infty))$.

There are two cases for $\psi_1$. If $\psi_1(0)=0$, then $\psi$ can be approximated as the limit of a sequence $\{\psi_{1,n}\}$ in $C_c^1(0,\infty)$. Arguing similarly as above, we obtain $\psi_1$ can be approximated in the $L^\infty$ norm by a sequence in $\mathcal{V}([0,\infty))$. If $\psi_1(0)=A\ne 0$, then $\psi(0)-A=0$, and therefore can be approximated in the $L^\infty$ norm by a sequence in $\mathcal{V}_0([0,\infty))\subset \mathcal{W}([0,\infty))$. Since $A\in \mathcal{W}([0,\infty))$, we deduce $\psi_1$  can be approximated in the $L^\infty$ norm by a sequence in $\mathcal{W}([0,\infty))$.

\end{proof}

\begin{proposition}\label{Propo:L4Boundedness}
For any $\psi\in\mathcal{W}([0,\infty))$, then define $\phi=\psi/|p|$, there exist two universal constants $c_1,c_2>0$ such that
$$|\mathfrak{H}^1_{\phi}(p_1,p_2)| \ \le c_1\ p_1p_2 \mbox{ 
and } |\mathfrak{H}^2_{\phi}(p,p)| \ \le c_2\ p^2$$
for $p_1\ge p_2\ge 0$ and $p\ge 0$.
\end{proposition}

\begin{proof}
For a vector $\psi\in\mathcal{W}([0,\infty))$, $\psi$ has the form
$A  \ + \ p\theta(p),$
where $A$ is a constant and $\theta \in \mathcal{V}([0,\infty))$. Since $\mathfrak{H}^1$ and $\mathfrak{H}^2$ are linear,
$\mathfrak{H}^1_\phi \ = \ \mathfrak{H}^1_{A/p} \ +  \ \mathfrak{H}^1_{\theta},
\mbox{ and }
\mathfrak{H}^2_\phi \ = \ \mathfrak{H}^2_{A/p} \ +  \ \mathfrak{H}^2_{\theta}.$ From Propositions \ref{Propo:L1Boundedness} and \ref{Propo:L2Boundedness}, it follows that
$|\mathfrak{H}^1_{\theta}(p_1,p_2)| \ \lesssim\ p_1p_2, \mbox{
and } |\mathfrak{H}^2_{\theta}(p,p)| \ \lesssim\ p^2.$ Let us now consider $\mathfrak{H}^1_{A/p}(p_1,p_2)$
$$\mathfrak{H}^1_{A/p}(p_1,p_2) \ =  \  2A|p_1+p_2|^2-2A(p_1^2+p_2^2)-4Ap_1p_2 +A|p_1-p_2|^2 \ = \ -4Ap_1p_2.$$
Moreover, it is also straightforward that
$\mathfrak{H}^2_{A/p}(p,p) \ =  \   -4Ap^2.$ As a consequence, the conclusion of the proposition follows.

\end{proof}
\subsection{The special effect of a  class of test functions on the collision operator}
In the propositions below, we show that $\mathfrak{H}^1_\phi$ and $\mathfrak{H}^2_\phi$ are positive with the test functions $\phi(p) \ = \ \phi_r(p) \ = \  \left(1-\frac{r}{p}\right)_+$, for $r>0$. This is the class of test functions that help us to detect the behavior of the solutions at infinity, due to the fact that the limit when $p$ tends to infinity of $\phi_r$ is $1$. 

\begin{proposition}[The positivity of $\mathfrak{H}^1_\phi$ and $\mathfrak{H}^2_\phi$]\label{Propo:L1L2Positivity}
For the special choices of $\phi$ belonging to the class
\begin{equation}
\label{Propo:L1L2Positivity}
\phi_{r}(p) \ = \ \left(1-\frac{r}{p}\right)_+, \ \ \ \ r\in(0,\infty),
\end{equation}
where
$$\left(1-\frac{r}{p}\right)_+ \ = \ 1-\frac{r}{p}, \mbox{ for } p\ge r,\ \ \mbox{ and  } \ \ \left(1-\frac{r}{p}\right)_+ \ = \ 0, \mbox{ for } 0\le p < r,$$
for $r\in(0,\infty)$, the two operators  $\mathfrak{H}^1_{\phi}(p_1,p_2), \mathfrak{H}^2_{\phi}(p_1,p_2)$ become non-negative for all $0\le p_2\le p_1<\infty$.  

In addition, $\phi$ satisfies the conditions of Propositions \ref{Propo:L1Boundedness} and \ref{Propo:L2Boundedness}:
 $\phi\in C([0,\infty])$, $p^2\phi'(p)\in L^\infty([0,\infty]),$ $\phi'(p)$ is piece-wise continuous on $(0,\infty)$ and $\exists \mathfrak{C}>0$ such that $p_2|\phi(p_1+p_2)-\phi(p_1)|\le\mathfrak{C}$ for $p_1\ge p_2\ge 0$.

Moreover, for this choice of test functions
$\mathfrak{H}^1_{\phi}(p,p) \ = \  \mathfrak{H}^2_{\phi}(p,p), \ \forall p\in[0,\infty),$
that means the weak formulation \eqref{Def:WeakSolution:E1} can be reformulated as
\begin{equation}\label{Def:WeakSolution:New}
\begin{aligned}
&\int_{[0,\infty]}\phi(t,p)g(t,p)|p|^2d \mu(p)-\int_{[0,\infty]}\phi(0,p)g(0,p)|p|^2d \mu(p)\\
\indent= &\int_0^t\bigg[\int_{[0,\infty]}\phi_s(s,p)g(s,p)|p|^2d \mu(p)+\int_{[0,\infty]^2}{g(s,p_1)|p_1|g(s,p_2)|p_2|}\mathfrak{H}^1_{\phi}(p_1,p_2)d \mu(p_1)d \mu(p_2)\bigg].
\end{aligned}
\end{equation}
\end{proposition}
\begin{proof} The proof is divided into two parts.

{\bf Part 1:} Positivity of $\mathfrak{H}^1_\phi$.

We rewrite the form of $\mathfrak{H}^1_\phi$ in terms of $\phi$ for $p_1\ge p_2\ge 0$
\begin{equation}\label{L1intermPhi}
\begin{aligned}
\mathfrak{H}^1_\phi(p_1,p_2) \ = & \ (p_1+p_2)^3\phi(p_1+p_2)-2(p_1^3+p_1p_2^2)\phi(p_1)-4p_1p_2^2\phi(p_2)+(p_1-p_2)^3\phi(p_1-p_2),
\end{aligned}
\end{equation}
and consider several cases.

{\it Case 1:} $p_1\ge p_2\ge p_1-p_2\ge r$. 

In this case, we compute
\begin{equation*}
\begin{aligned}
\mathfrak{H}^1_\phi(p_1,p_2)
\ = & \ -r(p_1+p_2)^{2}+2r\left(p_1^{2}+p_2^2\right) \ +4r p_1p_2-r(p_1-p_2)^{2}
\ \ge  \ 0.
\end{aligned}
\end{equation*}

{\it Case 2:} $p_1\ge p_2\ge r > p_1-p_2\ge 0$. 

Since $r > p_1-p_2$, it follows that $\phi(p_1-p_2) \ = \ 0 \ \ge \ 1-\frac{r}{(p_1-p_2)},$
which implies
\begin{equation*}
\begin{aligned}
\mathfrak{H}^1_\phi(p_1,p_2) \ \ge  & \ (p_1+p_2)^3\left(1-\frac{r}{p_1+p_2}\right)-2(p_1^3+p_1p_2^2)\left(1-\frac{r}{p_1}\right)\\
& \ -4p_1p_2^2\left(1-\frac{r}{p_2}\right)+(p_1-p_2)^3\left(1-\frac{r}{p_1-p_2}\right)
\ \ge  \ 0.
\end{aligned}
\end{equation*}

{\it Case 3:} $p_1\ge r >  p_1-p_2 ; p_2  \ge 0$. 

Since $\phi(p_2)=\phi(p_1-p_2)=0,$ it follows that
\begin{equation*}
\begin{aligned}
\mathfrak{H}^1_\phi(p_1,p_2) \ =  & \ (p_1+p_2)^3\left(1-\frac{r}{p_1+p_2}\right)-2(p_1^3+p_1p_2^2)\left(1-\frac{r}{p_1}\right).
\end{aligned}
\end{equation*}
Let us compute
\begin{equation*}
\begin{aligned}
2(p_1^3+p_1p_2^2) \ - \  (p_1+p_2)^3 \ =  & \ p_1^3 \ - \ 3p_1^2p_2 \ - \ p_1p_2^2 \ - \ p_2^3 \ \le (p_1-p_2)^3.
\end{aligned}
\end{equation*}
On the other hand, we have that
\begin{equation*}
\begin{aligned}
 (p_1+p_2)^3\left(-\frac{r}{p_1+p_2}\right)-2(p_1^3+p_1p_2^2)\left(-\frac{r}{p_1}\right) 
 =& \ r(p_1-p_2)^2.
\end{aligned}
\end{equation*}

Using the fact that $r>p_1-p_2$, we find 
$$(p_1+p_2)^3\left(-\frac{r}{p_1+p_2}\right)-2(p_1^3+p_1p_2^2)\left(-\frac{r}{p_1}\right) \ \ge \ 2(p_1^3+p_1p_2^2) \ - \  (p_1+p_2)^3,$$
which implies
$\mathfrak{H}^1_\phi(p_1,p_2) \ge 0.$

{\it Case 4:} $p_1< r$.  

In this case, it straightforward that $\mathfrak{H}^1_\phi(p_1,p_2)  \ = \ (p_1+p_2)^3\phi(p_1+p_2) \ \ge \ 0.$

{\bf Part 2:} Positivity of $\mathfrak{H}^2_\phi$.

Second, $\mathfrak{H}_\phi^2(p_1,p_2)$ can also be written as an operator of $\phi$
\begin{equation}
\label{L2intermPhi}
\begin{aligned}
\mathfrak{H}^2_\phi(p_1,p_2)
 \ = & \ (p_1+p_2)^2p_1[\phi(p_1+p_2) - \phi(p_1)] + (p_1+p_2)^2p_2[\phi(p_1+p_2) - \phi(p_2)].
\end{aligned}
\end{equation}
The monotonicity of $\phi$ gives $ \mathfrak{H}^2_\phi (p_1,p_2)\ge 0$ for all $p_1\ge p_2\ge 0$. 

\end{proof}

\section{Existence of weak solutions and  energy cascade}\label{Sec:ProofEnergyCascade}
\subsection{Existence of weak solutions}\label{Sec:Existence}
{In this section, we will show the existence of weak solutions in the sense of Definition 8. The proof is standard and is similar with the existence proof of Kierkels and Vel\'azquez  \cite{kierkels2015transfer}. It is the  classical regularized,  kernel cutting off strategy, commonly used for the  coagulation-fragmentation equation  (see, for instance, \cite{dubovski1996existence,saha2015singular}) and  the classical homogeneous Boltzmann equation (see, for instance, \cite{Arkeryd:1972:OBE}).}
\subsubsection{Regularized Equation}\label{Sec:RegularizedEquation}
The lemma below shows the existence of weak solutions for the regularized model. 

\begin{lemma}\label{Lemma:RegularizedEquationExistence}
Let $\varepsilon\in(0,1)$, $n\in \mathbb{N}$, and $g_0\in\mathfrak{D}([0,\infty])$,  $g_0\ge 0$ be arbitrary and
$$\int_{[0,\infty]}g_0(p)p^2 d \mu(p)<\infty.$$ Then there exists at least one function $g\in C([0,\infty):\mathfrak{D}([0,\infty]))$, $g \ge 0$, that for all $t\in[0,\infty)$ and all $\phi\in C^1([0,\infty):C([0,\infty]))$ satisfies
\begin{equation}\label{Lemma:RegularizedEquationExistence:1}
\begin{split}
&\int_{[0,\infty]}\phi(t,p)g(t,p)p^2 d \mu(p)-\int_{[0,\infty]}\phi(0,p)g_0(p)p^2 d \mu(p)\\
&\indent\begin{split}
=\int_0^t\bigg[&\int_{[0,\infty]}\phi_s(s,p)g(s,p)p^2 d \mu(p)\\&\indent+2\int_{p_1>p_2\ge0} {g(s,p_1)p_1^2g(s,p_2)p_2^2}\mathfrak{H}^{1,\epsilon,n}_{\phi(s,\cdot)}(p_1,p_2) d \mu(p_1) d \mu(p_2)\\
&\indent + \int_{p_1=p_2\ge0} {g(s,p_1)p_1^2g(s,p_2)p_2^2}\mathfrak{H}^{2,\epsilon,n}_{\phi(s,\cdot)}(p_1,p_2) d \mu(p_1) d \mu(p_2)\bigg] d s,
\end{split}
\end{split}
\end{equation}
where 
\begin{equation}\label{L1n}
\begin{aligned}
\mathfrak{H}^{1,\varepsilon,n}_\phi(p_1,p_2) \ = & \ \frac{1}{(p_1+\varepsilon)(p_2+\varepsilon)}\big[|p_1\wedge n+p_2\wedge n|^3\phi(p_1+p_2)+|p_1\wedge n-p_2\wedge n|^3\phi(p_1-p_2)\\
& -2((p_1\wedge n)^3+(p_1\wedge n)(p_2\wedge n)^2)\phi(p_1)-4(p_1\wedge n)(p_2\wedge n)^2\phi(p_2)\big],
\end{aligned}
\end{equation}
and
\begin{equation}\label{L2n}\begin{aligned}
\mathfrak{H}^{2,\varepsilon,n}_\phi(p_1,p_2) \ = & \ \frac{1}{(p_1+\varepsilon)(p_2+\varepsilon)}|p_1\wedge n+p_2\wedge n|^2[|p_1\wedge n+p_2\wedge n|\phi(p_1+p_2)\\
& -(p_1\wedge n)\phi(p_1)-(p_2\wedge n)\phi(p_2)\big],\end{aligned}
\end{equation}
with $a\vee	 b \ = \ \max\{a,b\}, \ \ \ \ 
a\wedge b \ = \ \min\{a,b\},\ \ \ \  a,b\in\mathbb{R}.$

Moreover, $$\int_{[0,\infty]}p^2g(t,p)d \mu(p)=\int_{[0,\infty]}p^2g(0,p)d \mu(p).$$

In the case that, we have in addition
$\int_{[0,\infty]}g_0(p)p d \mu(p)<\infty.$
The same result holds true except the test function $\phi$ belongs to  $\phi\in C^1([0,\infty):\mathfrak{M})$, where $\mathfrak{M}$ is defined in Definition \ref{Def:ContinuousSpaces}.
\end{lemma}

\begin{proof}
We only study the case when the condition $\int_{[0,\infty]}g_0(p)p^2 d \mu(p)<\infty,$
holds true. The other case can be done by a similar argument.

By setting  $g|p|^2$, $g_0|p|^2$ to be $g$, $g_0$, we reduce \eqref{Lemma:RegularizedEquationExistence:1} to a simpler equation
\begin{equation}\label{Lemma:RegularizedEquationExistence:E0}
\begin{split}
&\int_{[0,\infty]}\phi(t,p)g(t,p) d \mu(p)-\int_{[0,\infty]}\phi(0,p)g_0(p) d \mu(p)\\
&\indent\begin{split}
=\int_0^t\bigg[&\int_{[0,\infty]}\phi_s(s,p)g(s,p)d \mu(p)\\&\indent+2\int_{p_1>p_2\ge0} {g(s,p_1)g(s,p_2)}\mathfrak{H}^{1,\varepsilon,n}_{\phi(s,\cdot)}(p_1,p_2) d \mu(p_1) d \mu(p_2)\\
&\indent + \int_{p_1=p_2\ge0} {g(s,p_1)g(s,p_2)}\mathfrak{H}^{2,\varepsilon,n}_{\phi(s,\cdot)}(p_1,p_2) d \mu(p_1) d \mu(p_2)\bigg] d s.
\end{split}
\end{split}
\end{equation}
We then study \eqref{Lemma:RegularizedEquationExistence:E0} instead of \eqref{Lemma:RegularizedEquationExistence:1}.

{\bf Step 1: Local Existence.}

First, let us prove the local existence using a fixed point argument. To this end, we suppose that $g_0$ is non-zero, otherwise the proof is trivial. Now,  let $T\in(0,\infty)$  be determined later. Set $\omega$ to be a function in $C_c^\infty((-1,1))$  and  for  $\lambda\in(0,1)$ define $\omega_{\lambda}(x):=\frac1{\lambda}\omega\left(\frac{x}{\lambda}\right).$ 

Define the operator $\mathcal{O}_{\lambda}: C([0,T]:\mathfrak{D}([0,\infty]))\to C([0,T]:\mathfrak{D}([0,\infty]))$ 
in the following manner:  

For all $g\in C([0,T]:\mathfrak{D}([0,\infty]))$, all $t\in[0,T]$ and all $\phi\in C([0,\infty])$
\begin{equation}\label{Lemma:RegularizedEquationExistence:E1}\nonumber
\begin{split}
\int_{[0,\infty]}\phi(p)\mathcal{O}_{\lambda}[g](t,p)d \mu(p)&=\int_{[0,\infty]}\phi(p)g_0(p)e^{-\int_0^t\mathcal{P}_{\lambda}[g(s,\cdot)](p)d s}d \mu(p)\\
&\hspace{-32pt}+\int_0^t\int_{[0,\infty]}\phi(p)e^{-\int_s^t\mathcal{P}_{\lambda}[g(\sigma,\cdot)](p)d\sigma}\mathcal{Q}_{\lambda}[g(s,\cdot)](p)d \mu(p)\,d s,
\end{split}
\end{equation}
where $\mathcal{P}_{\lambda}:\mathfrak{D}([0,\infty])\rightarrow C_0([0,\infty))$ is defined as follows
\begin{equation}\label{Lemma:RegularizedEquationExistence:E2}\nonumber
\mathcal{P}_{\lambda}[g](p_1):=4\int_0^{p_1}((p_1\wedge n)^3+(p_1\wedge n)(p_2\wedge n)^2)\frac{(\omega_\lambda * g)(p_2)d \mu(p_2)}{{(p_1+\varepsilon)(p_2+\varepsilon)}},
\end{equation}
$$+\ 8\int_{-\infty}^\infty\int_{p_1}^\infty(p_1\wedge n)^2(p_2\wedge n) \frac{g(p_2)}{{(p_1+\varepsilon)(p_2+\varepsilon)}}\omega_\lambda(x-p_1)d \mu(p_2)d\mu(x).$$
and  $\mathcal{Q}_{\lambda}:\mathfrak{D}([0,\infty])\rightarrow\mathfrak{D}([0,\infty])$ is such that for all $g\in\mathfrak{D}([0,\infty])$ and all $\phi\in C([0,\infty])$ 
\begin{equation}\label{Lemma:RegularizedEquationExistence:E3}\nonumber
\int_{[0,\infty)}\phi(p)\mathcal{Q}_{\lambda}[g](p)d \mu(p)=\iint_{\{p_1>p_2\geq0\}}{g(p_1)(\omega_{\lambda}* g)(p_2)}\mathcal{K}(p_1,p_2)d \mu(p_1)d \mu(p_2),
\end{equation}
with
\begin{equation*}
\begin{aligned}
\mathcal{K}(p_1,p_2) \ = & \ \frac{2}{(p_1+\varepsilon)(p_2+\varepsilon)}\big[|p_1\wedge n+p_2\wedge n|^3\phi(p_1+p_2)+|p_1\wedge n-p_2\wedge n|^3\phi(p_1-p_2)\big].
\end{aligned}
\end{equation*}

It is clear that the two operators $\mathcal{P}_{\lambda}$ and $\mathcal{Q}_{\lambda}$ are well defined and  $\mathcal{O}_{\lambda}$ maps the space $C([0,T]:\mathfrak{D}([0,\infty]))$ into itself. Now, since $\mathcal{P}_{\lambda}$ and $\mathcal{Q}_{\lambda}$ are continuous operators on $\mathfrak{D}([0,\infty])$,  the operator $t\mapsto\mathcal{O}_{\lambda}[g](t,\cdot)$ is continuous on $[0,T]$.

Define the norms $$\|h\|_0=\sup_{\phi\in C([0,\infty]),\|\phi\|_{\infty}\leq1}\int_{[0,\infty]}\phi(p)h(p)d \mu(p) \mbox{ for } h\in\mathfrak{D}([0,\infty])$$
$$\mbox{  and } \|g\|_{*}=\sup_{t\in[0,T]}\|g(t,\cdot)\|_0 \mbox{ for } g\in C([0,T]:\mathfrak{D}([0,\infty])).$$

Taking into account the positivity of   $\mathcal{P}_{\lambda}$ for $g,g_0\ge0$, we find 
$\|\mathcal{O}_{\lambda}[g]\|_{*}\leq\|g_0\|_0+\tfrac{CTn^3}{\varepsilon^2}\|g\|_{*}^2,
$ for $g,g_0\ge0$ and some universal constant $C>0$. Now, let us define the fixed-point set $\mathfrak{X}_T:=\{g\in C([0,T]:\mathfrak{D}([0,\infty])):\|g\|_{*}\leq2\|g_0\|_0 \mbox{ and } g\ge 0\}.$
It is clear that $\mathfrak{X}_T$ is invariant under $\mathcal{O}_{\lambda}$ under the smallness condition of the time interval $T\leq\frac{\varepsilon^2}{Cn^3\|g_0\|_0}.$

Choosing  $g\in C([0,\infty):\mathfrak{D}([0,\infty]))$, and $t_1\in[0,T]$ and $t_2\in[t_1,T]$. Let $\phi\in C([0,\infty])$ with $\|\phi\|_\infty\leq1$. 
\begin{equation*}
\begin{aligned}
& \bigg|\int_{[0,\infty]}\phi(p)\mathcal{O}_{\lambda}[g](t_2,p)d \mu(p)-\int_{[0,\infty]}\phi(p)\mathcal{O}_{\lambda}[g](t_1,p)d \mu(p)\bigg|\\
\leq &\int_{t_1}^{t_2}\|\mathcal{Q}_{\lambda}[g(s,\cdot)]\|_*d s+\|g_0\|_0\,\left\|\int_{t_1}^{t_2}\mathcal{P}_{\lambda}[g(s,\cdot)](\cdot)d s\right\|_\infty\\
&+\int_0^{t_1}\|\mathcal{Q}_{\lambda}[g(s,\cdot)]\|_*d s\,\left\|\int_{t_1}^{t_2}\mathcal{P}_{\lambda}[g(s,\cdot)](\cdot)d s\right\|_\infty,
\end{aligned}
\end{equation*}

Now, suppose that $T\leq\frac{\varepsilon^2}{Cn^3\|g_0\|}$ and $g\in \mathfrak{X}_T$, the following holds true
\begin{equation*}
\begin{aligned}
& \bigg|\int_{[0,\infty]}\phi(p)\mathcal{O}_{\lambda}[g](t_2,p)d \mu(p)-\int_{[0,\infty]}\phi(p)\mathcal{O}_{\lambda}[g](t_1,p)d \mu(p)\bigg|\\
\le &\  C\left(\frac{n^3}{\varepsilon^2}\|g_0\|_0^2+\frac{Tn^6}{\varepsilon^4}\|g_0\|_0^3\right)\,\big|t_2-t_1\big|,
\end{aligned}
\end{equation*}
for some universal constant $C>0$.

By a classical argument, Arzel\`a-Ascoli theorem then implies that the operator  $\mathcal{O}_{\lambda}$ is a compact operator on $\mathfrak{X}_T$. Now, by Schauder's fixed point theorem there exists a fixed point $g_{\lambda}$ in the set $\mathfrak{X}_T$ such that $\mathcal{O}_{\lambda}[g_{\lambda}]\equiv g_{\lambda}$ on $[0,T]\times[0,\infty]$. As a consequence, $g_{\lambda}$  solves
\begin{equation*}
\begin{split}
\int_{[0,\infty]}\phi(p)g_{\lambda}(t,p)d \mu(p)&=\int_{[0,\infty]}\phi(p)g_0(p)e^{-\int_0^t\mathcal{P}_{\lambda}[g_{\lambda}(s,\cdot)](p)d s}d \mu(p)\\
&\hspace{-32pt}+\int_0^t\int_{[0,\infty]}\phi(p)e^{-\int_s^t\mathcal{P}_{\lambda}[g_{\lambda}(\sigma,\cdot)](p)d\sigma}\mathcal{Q}_{\lambda}[g_{\lambda}(s,\cdot)](p)d \mu(p)\,d s,
\end{split}
\end{equation*}
for  $\phi\in C([0,\infty])$ and  $t\in[0,T]$. In addition, we have
\begin{equation}\label{Lemma:RegularizedEquationExistence:E8aaa}
\begin{aligned}
\partial_t\bigg[{\int_{[0,\infty]}\phi(p)g_{\lambda}(t,p)d \mu(p)}\bigg] \ = & \ 2\iint_{\{p_1>p_2\geq0\}}{g_{\lambda}(t,p_1)(\omega_{\lambda}* g_{\lambda}(t,\cdot))(p_2)}\frac{1}{(p_1+\varepsilon)(p_2+\varepsilon)}\times\\
&\times \big[|p_1\wedge n+p_2\wedge n|^3\phi(p_1+p_2)+|p_1\wedge n-p_2\wedge n|^3\phi(p_1-p_2)\\
&-2((p_1\wedge n)^3+(p_1\wedge n)(p_2\wedge n)^2)\phi(p_1)\big]d \mu(p_1)d \mu(p_2)\\
&- 8\int_{-\infty}^\infty\iint_{\{p_2>p_1\geq0\}}{g_{\lambda}(t,p_2)\omega_{\lambda}(p-p_1) g_{\lambda}(t,p_1)}\frac{1}{(p_2+\varepsilon)(p_1+\varepsilon)}\times\\
&\times( p_2\wedge n) (p_1\wedge n)^2\phi(p_1)d\mu(p)d \mu(p_1)d \mu(p_2).
\end{aligned}
\end{equation}
Permuting the roles of $p_1$ and $p_2$ in the term containing $( p_2\wedge n) (p_1\wedge n)^2$ of the above equation, we obtain
\begin{equation}\label{Lemma:RegularizedEquationExistence:E8a}
\begin{aligned}
\partial_t\bigg[{\int_{[0,\infty]}\phi(p)g_{\lambda}(t,p)d \mu(p)}\bigg] 
\ = & \ 2\iint_{\{p_1>p_2\geq0\}}{g_{\lambda}(t,p_1)(\omega_{\lambda}* g_{\lambda}(t,\cdot))(p_2)}\frac{1}{(p_1+\varepsilon)(p_2+\varepsilon)}\times\\
&  \times \big[|p_1\wedge n+p_2\wedge n|^3\phi(p_1+p_2)+|p_1\wedge n-p_2\wedge n|^3\phi(p_1-p_2)\\
& -2((p_1\wedge n)^3+(p_1\wedge n)(p_2\wedge n)^2)\phi(p_1)\big]d \mu(p_1)d \mu(p_2)\\
& - 8\iint_{\{p_1>p_2\geq0\}}{g_{\lambda}(t,p_1)(\omega_{\lambda}* g_{\lambda}(t,\cdot))(p_2)}\frac{1}{(p_1+\varepsilon)(p_2+\varepsilon)}\times\\
&  \times( p_1\wedge n) (p_2\wedge n)^2\phi(p_2)d \mu(p_1)d \mu(p_2).
\end{aligned}
\end{equation}
Using the fact that $(\omega_{\lambda} * g_{\lambda}(s,\cdot))$ is smooth, then the integral over $\{p_1=p_2\geq0\}$ of  $g_{\lambda}(s,\cdot)\times(\omega_{\lambda} * g_{\lambda}(s,\cdot))(\cdot)$ is zero and hence the second term on the right hand side is zero
\begin{equation}\label{Lemma:RegularizedEquationExistence:E8}
\begin{aligned}
\partial_t\bigg[{\int_{[0,\infty]}\phi(p)g_{\lambda}(t,p)d \mu(p)}\bigg]
\ =  \  & 2\iint_{\{p_1>p_2\geq0\}}{g_{\lambda}(t,p_1)(\omega_{\lambda}* g_{\lambda}(t,\cdot))(p_2)}\frac{1}{(p_1+\varepsilon)(p_2+\varepsilon)}\times\\
&  \times \big[|p_1\wedge n+p_2\wedge n|^3\phi(p_1+p_2)+|p_1\wedge n-p_2\wedge n|^3\phi(p_1-p_2)\\
& -2((p_1\wedge n)^3+(p_1\wedge n)(p_2\wedge n)^2)\phi(p_1)\big]d \mu(p_1)d \mu(p_2)\\
&- 8\iint_{\{p_1>p_2\geq0\}}{g_{\lambda}(t,p_1)(\omega_{\lambda}* g_{\lambda}(t,\cdot))(p_2)}\frac{1}{(p_1+\varepsilon)(p_2+\varepsilon)}\times\\
& \times( p_1\wedge n) (p_2\wedge n)^2\phi(p_2)d \mu(p_1)d \mu(p_2)\\
&+ \ \int_{\{p_1=p_2\ge0\}}{g_{\lambda}(t,p_1)(\omega_{\lambda}* g_{\lambda}(t,\cdot))(p_2)}\mathfrak{H}^{2,\epsilon,n}_{\phi}(p_1,p_2) d \mu(p_1) d \mu(p_2).
\end{aligned}
\end{equation}

Equation \eqref{Lemma:RegularizedEquationExistence:E8} is then equivalent with
\begin{equation}\label{Lemma:RegularizedEquationExistence:E9}
\begin{split}
&\int_{[0,\infty]}\phi(p) g_{\lambda}(t,p)d \mu(p)-\int_{[0,\infty]}\phi(p)g_0(p)d \mu(p)\\
\hspace{8pt}=\ &\ \int_0^t\iint_{[0,\infty]^2}{g_{\lambda}(s,p_1\vee p_2)(\omega_{\lambda} * g_{\lambda}(s,\cdot))(p_1\wedge p_2)}\frac{1}{(p_1+\varepsilon)(p_2+\varepsilon)}\times\\
&  \times \big[|p_1\wedge n+p_2\wedge n|^3\phi(p_1+p_2)+|p_1\wedge n-p_2\wedge n|^3\phi(p_1-p_2)\\
& -2((p_1\wedge n)^3+(p_1\wedge n)(p_2\wedge n)^2)\phi(p_1)\big]d \mu(p_1)d \mu(p_2)\,d s\\
&- 4\int_0^t\iint_{[0,\infty]^2}{g_{\lambda}(t,p_1\vee p_2)(\omega_{\lambda}* g_{\lambda}(t,\cdot))(p_1\wedge p_2)}\frac{1}{(p_1+\varepsilon)(p_2+\varepsilon)}\times\\
& \times( p_1\wedge n) (p_2\wedge n)^2\phi(p_2)d \mu(p_1)d \mu(p_2)ds.
\end{split}
\end{equation}

By choosing $\phi\equiv1$ in \eqref{Lemma:RegularizedEquationExistence:E9}, we find the conservation of energy
$$\int_{[0,\infty]} g_{\lambda}(t,p)d \mu(p)=\int_{[0,\infty]}g_0(p)d \mu(p)$$
 for all $t\in[0,T]$.

{\bf Step 2: Global Existence.}

We will use the classical argument to prove that the regularized equation has a global solution. Choosing the initial datum $\tilde{g}_0:=g(T,\cdot)$, applying again the argument of Step 1 to \eqref{WeakTurbulenceInitial} in the interval $[T,2T]$, we can find a solution $g_{\lambda}\in C([T,2T]:\mathfrak{D}([0,\infty]))$ so that
\begin{equation}\label{Lemma:RegularizedEquationExistence:E10}
\begin{split}
\int_{[0,\infty]}\phi(p)g_{\lambda}(t,p)d \mu(p)&=\int_{[0,\infty)}\phi(p)g_{\lambda}(T,p)e^{-\int_T^t\mathcal{P}_{\lambda}[g_{\lambda}(s,\cdot)](p)d s}d \mu(p)\\
&\hspace{-32pt}+\int_T^t\int_{[0,\infty)}\phi(p)e^{-\int_s^t\mathcal{P}_{\lambda}[g_{\lambda}(\sigma,\cdot)](p)d\sigma}\mathcal{Q}_{\lambda}[g_{\lambda}(s,\cdot)](p)d \mu(p)\,d s,
\end{split}
\end{equation}
for all $\phi\in C([0,\infty])$ and all $t\in[T,2T]$.

Note that from Step 1, we observe that $T$ depends only on $n,\varepsilon$ and $\lambda$ since the energy is conserved $\|\tilde{g}_0\|_{L^1}=\|g_0\|_{L^1}.$ Therefore, the existence of the solution $g_\lambda$ on $[T,2T]$ is guaranteed. Repeating this argument on $[2T,3T],[3T,4T],[4T,5T],\dots$ we obtain for any $\lambda\in(0,1)$ a global weak solution $g_{\lambda}\in C([0,\infty):\mathfrak{D}([0,\infty]))$ of \eqref{WeakTurbulenceInitial} for all $\phi\in C([0,\infty])$ and all $t\in[0,\infty)$.

{\bf Step 3: The limit $\lambda \to 0$.}

We  will now take the limit $\lambda\rightarrow0$. Since $\{g_\lambda\}_{\lambda\in(0,1)}$ is bounded, as well as equicontinuous, by Arzel\`a-Ascoli theorem, we deduce that  the family $\{g_\lambda\}_{\lambda\in(0,1)}$ is precompact in $C([0,\infty):\mathfrak{D}([0,\infty]))$.  
Therefore, there exists  $g\in C([0,\infty):\mathfrak{D}([0,\infty]))$, such that
 $\int_{[0,\infty]}g(t,p)d \mu(p)=\|g_0\|_{L^1}$ for all $t\in[0,\infty)$, and a sequence ${\lambda}_m\rightarrow0$ such that $g_{\lambda_m}\stackrel{\ast}{\rightharpoonup} g$  uniformly and locally in $t$ on $[0,\infty)$. 
 Therefore, for a time independent test function $\phi$, the left hand side of \eqref{Lemma:RegularizedEquationExistence:E8}  converges  to the left hand side of \eqref{Lemma:RegularizedEquationExistence:E0} and $\omega_{\lambda_m} *  g_{\lambda_m}(t,\cdot)\stackrel{\ast}{\rightharpoonup} g(t,\cdot),$ locally uniformly in $t$ on $[0,\infty)$.

As a result, the right hand side of \eqref{Lemma:RegularizedEquationExistence:E8} converges to
\begin{equation}\label{Lemma:RegularizedEquationExistence:E11}\nonumber
\begin{aligned}
& \iint_{[0,\infty]^2}{g(t,p_1)g(t,p_2)}\mathfrak{H}^{1,\varepsilon,n}_{\phi}(p_1,p_2)d \mu(p_1)d \mu(p_2)\,\\
=& \indent 2\int_{p_1>p_2\ge0} {g(t,p_1)g(t,p_2)}\mathfrak{H}^{1,\varepsilon,n}_{\phi}(p_1,p_2) d \mu(p_1) d \mu(p_2)\\
& + \int_{p_1=p_2\ge0} {g(t,p_1)g(t,p_2)}\mathfrak{H}^{2,\varepsilon,n}_{\phi}(p_1,p_2) d \mu(p_1) d \mu(p_2),
\end{aligned}
\end{equation} in which  with $\phi$ is time independent.

When the test function is time dependent $\phi\in C^1([0,\infty):C([0,\infty]))$,  the  linear term in  \eqref{Lemma:RegularizedEquationExistence:E8} will appear. The linear term will converge to the first term on the right hand side of \eqref{Lemma:RegularizedEquationExistence:E0}. As a consequence, the  function $g$ will satisfy \eqref{Lemma:RegularizedEquationExistence:E0} for all $t\in[0,\infty)$ and all $\phi\in C^1([0,\infty):C([0,\infty]))$.

\end{proof}
\subsubsection{Energy at \{p=0\}}\label{Sec:MassInfinity}

In this section we will prove that $\int_{\{0\}}g(t,p)p^2d \mu(p)=0$ for all $t\in(0,\infty)$. The following estimate is true independently of $\varepsilon\in(0,1)$, $n\in\mathbb{N}$ and will allow us to take the limit $\varepsilon\rightarrow0$, $n\rightarrow\infty$.
\begin{lemma}\label{Lemma:MassInfinity}
Let  $g_0\in\mathfrak{D}([0,\infty])$. Suppose that $gp^2\in C([0,\infty):\mathfrak{D}([0,\infty]))$ ,  $g\ge0$ satisfies \eqref{Lemma:RegularizedEquationExistence:1} for all $t\in[0,\infty)$. Then given $\rho\in (0,1),R\in(0,\infty)$ the following holds true
\begin{equation}\label{Lemma:MassInfinity:1}
\int_{[R\rho,\infty]}g(t,p)p^2d \mu(p)\geq(1-\rho)\int_{[R,\infty]}g_0(p)p^2d \mu(p)\text{ for all }t\in[0,\infty).
\end{equation}
\end{lemma}

\begin{proof}
For  $\rho\in(0,1)$, fix $\kappa\in(0,R)$ and set $\phi(p)=\phi_\kappa(p)=\left(1-\frac{\kappa}{ p}\right)_+$ defined in Proposition \ref{Propo:L1L2Positivity}, then 
\begin{equation}\label{Lemma:MassInfinity:E1}
\begin{split}
\int_{[\kappa,\infty]}g(t,p)p^2d \mu(p)&\geq\int_{[0,\infty]}\left(1-\frac{\kappa}{ p}\right)_+g(t,p)p^2d \mu(p)\geq\int_{[0,\infty]}\left(1-\frac{\kappa}{ p}\right)_+g_0(p)p^2d \mu(p)\\&\geq\left(1-\frac{\kappa }{R}\right)\int_{[R,\infty]}g_0(p)p^3d \mu(p)\text{ for all }t\in[0,\infty).
\end{split}
\end{equation}
Choosing $\kappa=\rho R$, we obtain the conclusion of the lemma.
\end{proof}
\begin{corollary}\label{cor:nomassatinfty}
Let $g_0p^2 \in\mathfrak{D}([0,\infty])$. Suppose that $gp^2\in C([0,\infty):\mathfrak{D}([0,\infty]))$,  $g\ge0$ satisfies \eqref{Lemma:RegularizedEquationExistence:1} for all $t\in[0,\infty)$ and all $\phi\in C^1([0,\infty):C([0,\infty]))$. 
Then  ~$\int_{\{0\}}g(t,p)p^2 d \mu(p)\equiv0$ for all $t\in[0,\infty)$.

\end{corollary}

\begin{proof}
Using $\phi(p)\equiv 1$ in \eqref{Lemma:RegularizedEquationExistence:1}, it follows that
$\int_{[0,\infty]}g(t,p)p^2d \mu(p)=\|g_0p^2\|_1$ for all $t\in[0,\infty)$.
Moreover, for any $t\in[0,\infty)$ if follows from Lemma \ref{Lemma:MassInfinity} that
\begin{equation}\nonumber
\|g_0p^2\|_1=\int_{[0,\infty]}g(t,p)p^2d \mu(p)\geq\int_{\left[\frac{1}{R^2},\infty\right]}g(t,p)p^2d \mu(p)\geq\frac{R-1}{R}\int_{\left[\frac{1}{R},\infty\right]}g_0(p)p^2d \mu(p).
\end{equation}
The right hand side  tends to $\|g_0p^2\|_1$ as $R\rightarrow\infty$, we then find $\int_{\{0\}}g(t,p)p^2d \mu(p)=0$ for all $t\in[0,\infty)$. 
\end{proof}

\subsubsection{Existence result}\label{Sec:ExistenceResult}
In this proposition, we will show the existence of weak solutions in the sense of Definition \ref{Def:WeakSolution}, passing the solutions of the regularized model to the limit. 

\begin{proposition}[Existence of weak solutions]\label{Propo:Existence}
Given any $g_0p^2\in\mathfrak{D}([0,\infty])$, $g_0\ge0$ and
$$\int_{[0,\infty]}g_0(p)p^2d\mu(p)<\infty$$ there exists at least one weak solution $g\in C([0,\infty):\mathfrak{D}([0,\infty]))$,  $g\ge0$  in the sense of Definition \ref{Def:WeakSolution}  that satisfies $g(0,\cdot)=g_0$. Moreover,
$$\int_{[0,\infty]}g(t,p)p^2d\mu(p)\ = \  \int_{[0,\infty]}g_0(p)p^2d\mu(p)\ <\ \infty, \mbox{ for } t\ge 0.$$

If in addition, 
$$\int_{[0,\infty]}g_0(p)pd\mu(p)<\infty$$ there exists at least one weak solution $g\in C([0,\infty):\mathfrak{D}([0,\infty]))$,  $g\ge0$ to \eqref{WeakTurbulenceEnergy} that satisfies $g(0,\cdot)=g_0$. Moreover,
$$\int_{[0,\infty]}g(t,p)p^2d\mu(p)\ = \  \int_{[0,\infty]}g_0(p)p^2d\mu(p)\ <\ \infty,  \mbox{ for } t\ge 0,$$
$$\int_{[0,\infty]}g(t,p)pd\mu(p)\ \le  \  \int_{[0,\infty]}g_0(p)pd\mu(p)\ <\ \infty, \mbox{ for } t\ge 0.$$

\end{proposition}

\begin{proof} We consider first the case 
$$\int_{[0,\infty]}g_0(p)p^2dp<\infty.$$
From the previous propositions, we know that for any $\varepsilon\in(0,1)$, $n\in\mathbb{N}$ there exists a solution $g_{\varepsilon,n}\in C([0,\infty):\mathfrak{D}([0,\infty)))$ of \eqref{Lemma:RegularizedEquationExistence:1} for all $t\in[0,\infty)$ and all $\phi\in C^1([0,\infty):C([0,\infty]))$. 

Let us denote the collection of these solutions by $\mathfrak{S}=\{g_{\varepsilon,n}\}_{\varepsilon\in(0,1),n\in\mathbb{N}}.$ For any test function $\phi\in \mathcal{V}([0,\infty]) \subset C^1([0,\infty])$  and any times $t_1,t_2\in[0,\infty)$, equation \eqref{Lemma:RegularizedEquationExistence:1} implies the following estimate for any $g_{\varepsilon\in(0,1),n\in\mathbb{N}}\in\mathfrak{S}$, 
\begin{equation}\label{Lipschitz}
\bigg|\int_{[0,\infty]}\phi(p)g_\varepsilon(t_2,p)p^2 d \mu(p)-\int_{[0,\infty]}\phi(p)g_\varepsilon(t_1,p)p^2 d \mu(p)\bigg|\leq|t_2-t_1|C_{\phi,\varepsilon,n}\|g_0p^2\|_{L^1}^2,
\end{equation}
where, due to Proposition \ref{Propo:L4Boundedness},
$
C_{\phi,\varepsilon,n}  =  C\left[\sup_{p_1>p_2\ge0 } \frac{\mathfrak{H}^1_\phi(p_1,p_2)}{p_1p_2}\ + \ \sup_{p\ge 0 } \frac{\mathfrak{H}^2_\phi(p,p)}{p^2}\right]\|\phi\|_\mathcal{V},
$ in which $C$ is a universal constant.

Now, since $\mathcal{V}([0,\infty])$ is dense in $C([0,\infty])$ according to Proposition \ref{Propo:L3Boundedness}, for any $\phi\in C([0,\infty])$ the family of mappings $t\mapsto\int_{[0,\infty)}\phi(x)g_{\varepsilon,n}(t,p)p^2d \mu(p)$ with $g_{\varepsilon,n}\in\mathfrak{S}$, is uniformly continuous on $[0,\infty)$.

As a consequence, applying Arzel\`a-Ascoli theorem again, we find that the family $\mathfrak{S}$ is precompact in $C([0,\infty):\mathfrak{D}([0,\infty]))$. By Corollary \ref{cor:nomassatinfty} there  exist sequences $\varepsilon_m\rightarrow0, n_m\to\infty$ and some function $g\in C([0,\infty):\mathfrak{D}([0,\infty]))$ such that $\|g(t,\cdot)p^2\|_{L^1}=\|g_0p^2\|_{L^1}$ for all $t\in[0,\infty)$ and $g_{\varepsilon_m,n_m}(t,p)p^2\stackrel{\ast}{\rightharpoonup} g(t,p)p^2$ locally uniformly in $t$ on $[0,\infty)$. Using the definition of of weak$^*$ convergence, we infer that the left hand side and the first term on the right hand side of \eqref{Lemma:RegularizedEquationExistence:1} converge to the corresponding terms in \eqref{WeakTurbulenceInitial}. 
In addition, it is also clear that $g_{\varepsilon_m,n_m}(s,p_1)p_1g_{\varepsilon_m,n_m}(s,p_2)p_2\stackrel{\ast}{\rightharpoonup} g(s,p_1)p_1g(s,p_2)p_2$ in $\mathfrak{D}([0,\infty]^2)$, locally uniformly in $s\in[0,\infty)$. Moreover, we find
\begin{equation*}
\begin{aligned}
\mathfrak{H}^{1,\varepsilon,n}_\phi(p_1,p_2)
\longrightarrow & \ \ \ \  \frac{1}{p_1p_2}\big[|p_1+p_2|^3\phi(p_1+p_2)+|p_1-p_2|^3\phi(p_1-p_2)\\
& -2(p_1^3+p_1p_2^2)\phi(p_1)-4p_1p_2^2\phi(p_2)\big]
\end{aligned}
\end{equation*} uniformly in $\{p_1\geq p_2\geq 0\}$ due to Proposition \ref{Propo:L4Boundedness}. The conclusion of the Proposition then follows.

Now, if in addition
$\int_{[0,\infty]}g_0(p)pd\mu(p)<\infty,$ 
then the above argument still holds true, except $C([0,\infty])$ is replaced by $\mathfrak{M}$ (see Definition \ref{Def:ContinuousSpaces}). And in inequality  \eqref{Lipschitz}, the constant $C_{\phi,\varepsilon,n}$ depends on the constants $c_1, c_2$ of Proposition \ref{Propo:L4Boundedness}. By choosing the test function $\phi=1/p$, we get
$$\int_{[0,\infty]}g(t,p)pd\mu(p)\ \le  \  \int_{[0,\infty]}g_0(p)pd\mu(p)\ <\ \infty, \mbox{ for } t\ge 0.$$

\end{proof}

\subsection{Monotonicity of the energy at $\{\infty\}$ and trivial solutions}\label{Sec:Monoton}
The following proposition indicates that the energy  at $\{\infty\}$ is indeed non-decreasing.
\begin{proposition}[Monotonicity of the measure of $\{\infty\}$]\label{Propo:Monoton}
Let $g$, $gp^2\in C([0,\infty):\linebreak\mathfrak{D}([0,\infty]))$,  $g\ge0$ be a weak solution in the sense of Definition \ref{Def:WeakSolution}. Then the mapping $$t\mapsto\int_{\{\infty\}}g(t,p)p^2d \mu(p)$$ is nondecreasing on $[0,\infty)$.
\end{proposition}
\begin{proof}
Let $\nu\in\mathfrak{D}([0,\infty))$, the following holds true  $$\int_{\{\infty\}}\nu(p)p^3 d \mu(p)=\inf_{\phi_r}\int_{[0,\infty)}\phi_r(p)\nu(p)p^3d \mu(p),$$
in which the infimum is taken over all increasing functions $\phi_r\in C([0,\infty])$  defined in Proposition \ref{Propo:L1L2Positivity}. 

For any of these test functions, it is clear that $\mathfrak{H}^1_{\phi_r(p)},\mathfrak{H}^2_{\phi_r(p)} \geq0$. As a consequence, the mappings $t\mapsto\int_{[0,\infty]}\phi_r(p)g(t,p)p^2d \mu(p)$ are nondecreasing on $[0,\infty)$. 
Therefore $t\mapsto\int_{\{\infty\}}g(t,p)p^2d \mu(p)$ is also nondecreasing on $[0,\infty)$, due to the fact that it is the infimum of a collection of nondecreasing functions.
\end{proof}

A consequence of the previous proposition is the following result.
\begin{corollary}\label{cor:trivialsolutions}[Stationary (trivial) solutions]
Suppose that $g\in C([0,\infty):\mathfrak{D}([0,\infty]))$ is a weak solution  in the sense of Definition \ref{Def:WeakSolution}  satisfying $g(0,p)p^2\equiv E\delta_{\{p=\infty\}}$ for some $E\in[0,\infty)$. Then $g$ is a trivial solution of \eqref{WeakTurbulenceEnergy} in the sense $g(t,p)p^2\equiv E\delta_{\{p=\infty\}}$ for all $t\in[0,\infty)$.
\end{corollary}
\begin{remark} In the context of quantum Boltzmann equations, $\delta_{\{p=0\}}$ is the trivial equilibrium \cite{Lu:2014:TBE}.
\end{remark}
\begin{proof}
Using the previous proposition
\begin{equation}\nonumber
\int_{\{p=\infty\}}g(0,p)p^2d \mu(p) = E\leq \int_{\{p=\infty\}}g(t,p)p^2d \mu(p)\text{ for all }t\in[0,\infty),
\end{equation}
which implies
\begin{equation}\nonumber
0\leq\int_{[0,\infty)}g(t,p)p^2d \mu(p)=E-\int_{\{p=\infty\}}g(t,p)p^2d \mu(p)\leq0\text{ for all }t\in[0,\infty).
\end{equation}

Therefore, $g(t,p)\equiv0$ on $[0,\infty)\times[0,\infty)$ and hence $g$ is trivial.
\end{proof}

\subsection{Conservation of energy}\label{Sec:Conservation}
\begin{proposition}[Conservation of energy]\label{Propo:Conservation}
Suppose that $gp^2\in C([0,\infty):\mathfrak{D}([0,\infty]))$,  $g\ge0$ is a weak solution  in the sense of Definition \ref{Def:WeakSolution}, then
\begin{equation}\label{Propo:Conservation:2}
\int_{[0,\infty]}p^2g(t,p)d \mu(p)=\int_{[0,\infty]}p^2g(0,p)d \mu(p)\text{ for all }t\in[0,\infty).
\end{equation}
\end{proposition}\begin{proof}
\begin{remark}
\end{remark}

\eqref{Propo:Conservation:2} follows immediately from Lemma \ref{Lemma:WeakFormulation} by choosing the test function $\phi\equiv1$.

\end{proof}

\subsection{Cascade and accumulation of energy toward $\{\infty\}$}\label{Sec:LongTime}
The following two lemmas prove  that the energy is accumulated at $\{\infty\}$. The first lemma compares the energy on $[R\rho,\infty]$ and $[R,\infty]$ for some numbers $R$ and $\rho$ at two different  times. The second one compares the energy between $[r,\infty]$ and $[0,r]$.

\begin{lemma}[Cascade of energy toward $\{\infty\}$]\label{Propo:Tightness}
Let $gp^2\in C([0,\infty):\mathfrak{D}([0,\infty]))$,  $g\ge0$ be a non-trivial weak solution  in the sense of Definition \ref{Def:WeakSolution} . Then given $\rho\in (0,1)$, $R\in(0,\infty)$ the following holds.
\begin{equation}\label{Propo:Tightness:1}
\int_{[R \rho, \infty]}g(t,p)p^2d \mu(p)\geq(1-\rho)\int_{[R,\infty]}g(0,p)p^2d \mu(p)\text{ for all }t\in[0,\infty)
\end{equation}
\end{lemma}
\begin{proof} The proof is the same as the one of Lemma \ref{Lemma:MassInfinity}. 
\end{proof}
\begin{corollary}\label{cor:nomassatinftynew}
Let $gp^2\in C([0,\infty):\mathfrak{D}([0,\infty]))$,  $g\ge0$ be a non-trivial weak solution in the sense of Definition \ref{Def:WeakSolution}. Then if,
$$\int_{\{0\}}g(0,p)p^2d \mu(p)\equiv0,$$
we have
$$\int_{\{0\}}g(t,p)p^2d \mu(p)\equiv0$$ for all $t\in[0,\infty)$. 
\end{corollary}
\begin{proof} The proof is the same as the one of Corollary \ref{cor:nomassatinfty}. 
\end{proof}
\begin{lemma}[Accumulation of energy toward $\{\infty\}$]\label{lemma:KeyLemma}
Suppose that $g$, $gp^2\in C([0,\infty):\mathfrak{D}([0,\infty]))$,  $g\ge0$ is a  non-trivial weak solution  in the sense of Definition \ref{Def:WeakSolution}. For any $r\in(0,\infty)$ and all $t\in[0,\infty)$ the following holds.
\begin{equation}\label{lemma:KeyLemma:1}
\int_{[r,\infty]}g(t,p)p^2d \mu(p)\geq2\int_0^t\iint_{[0,r]^2}{g(s,p_1)p_1^2g(s,p_2)p_2^2}{[p_1p_2]^\frac12}\left(1-\frac{r}{p_1+p_2}\right)_+d \mu(p_1)d \mu(p_2)\ d s.
\end{equation}
\end{lemma}
\begin{proof}

Choose the test function $\phi$ to be $\phi_r$ defined in Proposition \ref{Propo:L1L2Positivity},  we find

\begin{equation}\label{lemma:KeyLemma:E1}
\begin{aligned}
& \int_{[0,\infty]}g(t,p)p^2\phi_r(p)d\mu(p) - \int_{[0,\infty]}g(0,p)p^2\phi_r(p)d\mu(p) \\
 =& \ 2\int_0^t\int_{p_1>p_2\ge0}{|p_1||p_2|}g(s,p_1)g(s,p_2)\mathfrak{H}^1_{\phi_r}(p_1,p_2){d}\mu(p_{1})\,{d}\mu(p_2)ds\,\\
 &  + \int_0^t\int_{p_1=p_2\ge0}{|p_1||p_2|}g(s,p_1)g(s,p_2)\mathfrak{H}^2_{\phi_r}(p_1,p_2) {d}\mu(p_{1})\,{d}\mu(p_2)ds\,.
\end{aligned}
\end{equation}
which, by the fact that $g_0\ge 0$, yields
\begin{equation}\label{lemma:KeyLemma:E2}
\begin{aligned}
 \int_{[0,\infty]}g(t,p)p^2\phi_r(p)d\mu(p) \ \ge & \ 2\int_0^t\int_{ p_1>p_2\ge0}{|p_1||p_2|}g(s,p_1)g(s,p_2)\mathfrak{H}^1_{\phi_r}(p_1,p_2){d}\mu(p_{1})\,{d}\mu(p_2)ds\,\\
 & \ + \int_0^t\int_{p_1=p_2\ge0}{|p_1||p_2|}g(s,p_1)g(s,p_2)\mathfrak{H}^2_{\phi_r}(p_1,p_2) {d}\mu(p_{1})\,{d}\mu(p_2)ds\,.
\end{aligned}
\end{equation}

Since $\mathfrak{H}^1_{\phi_r}(p_1,p_2), \mathfrak{H}^2_{\phi_r}(p_1,p_2) \ge 0$, we can restrict the integrals on $p_1>p_2\ge0$ and $p_1=p_2\ge0$ to $r\ge p_1>p_2\ge0$ and $r\ge p_1=p_2\ge0$, yielding
\begin{equation}\label{lemma:KeyLemma:E2aa}
\begin{aligned}
 \int_{[0,\infty]}g(t,p)p^2\phi_r(p)d\mu(p) \ \ge & \ 2\int_0^t\int_{r\ge p_1>p_2\ge0}{|p_1||p_2|}g(s,p_1)g(s,p_2)\mathfrak{H}^1_{\phi_r}(p_1,p_2){d}\mu(p_{1})\,{d}\mu(p_2)ds\,\\
 & \ + \int_0^t\int_{r\ge p_1=p_2\ge0}{|p_1||p_2|}g(s,p_1)g(s,p_2)\mathfrak{H}^2_{\phi_r}(p_1,p_2) {d}\mu(p_{1})\,{d}\mu(p_2)ds\,.
\end{aligned}
\end{equation}
We compute $\mathfrak{H}^1_{\phi_r}(p_1,p_2) $ when $r\ge p_1>p_2\ge0$
\begin{equation}\label{L1intermPhi1}
\begin{aligned}
\mathfrak{H}^1_{\phi_r}(p_1,p_2) 
 \ = & \ p_1^3[\phi_r(p_1+p_2)+\phi_r(p_1-p_2)-2\phi_r(p_1)]+3p_1^2p_2[\phi_r(p_1+p_2)-\phi_r(p_1-p_2)]\\
 & +p_1p_2^2[3\phi_r(p_1+p_2)-6\phi_r(p_1)+3\phi_r(p_1-p_2)]\\
  & +p_1p_2^2[4\phi_r(p_1)-4\phi_r(p_2)]\    + \ p_2^3[\phi_r(p_1+p_2)-\phi_r(p_1-p_2)].
\end{aligned}
\end{equation}
Since $\phi$ is increasing, it is clear from \eqref{L1intermPhi1} that
\begin{equation}\label{lemma:KeyLemma:E2a}
\begin{aligned}
\mathfrak{H}^1_{\phi_r}(p_1,p_2) 
 \ \ge & \  (p_1^3+3p_1p_2^2)[{\phi_r}(p_1+p_2)+{\phi_r}(p_1-p_2)-2{\phi_r}(p_1)].
\end{aligned}
\end{equation}
Since  ${\phi_r}(p_1+p_2)+{\phi_r}(p_1-p_2)-2{\phi_r}(p_1)={\phi_r}(p_1+p_2)$ for $p_1\leq r$, we then compute
 $${\phi_r}(p_1+p_2)+{\phi_r}(p_1-p_2)-2{\phi_r}(p_1)\geq{\phi_r}(p_1+p_2)\mathbf{1}_{[0,r]}(p_1)$$ on $\{r\ge p_1\geq p_2\geq 0\}$. In addition, $\mathfrak{H}^2_\phi(p,p)=\mathfrak{H}^1_\phi(p,p)$. 
 
Combining the above estimates yields
\begin{equation}\label{lemma:KeyLemma:E2b}
\begin{aligned}
\int_{\mathbb{R}^3}g(t,p)p^2\phi(p)d\mu(p) 
 \ge & \ \int_0^t\iint_{[0,r]^2}{g(s,p_1)p_1g(s,p_2)p_2}(p_1^3+3p_1p_2^2)\phi(p_1+p_2)d \mu(p_1)d \mu(p_2)d s \\
 \ge & \ 2\int_0^t\iint_{[0,r]^2}{g(s,p_1)p_1^2g(s,p_2)p_2^2}(p_1p_2)^{
\frac12}\phi(p_1+p_2)d \mu(p_1)d \mu(p_2)d s.
\end{aligned}
\end{equation}
The above inequalities  lead to
\begin{equation}\label{lemma:KeyLemma:E3}
\int_{[r,\infty]}\,g(t,p)p^2d \mu(p)\geq 2\int_0^t\iint_{[0,r]^2}{g(s,p_1)p_1^2g(s,p_2)p_2^2}{(p_1p_2)^{
\frac12}}\phi(p_1+p_2)d \mu(p_1)d \mu(p_2)d s,
\end{equation}
then \eqref{lemma:KeyLemma:1} holds true for all $t\in[0,\infty)$.
\end{proof}

\subsection{Transferring of energy away from $\{0\}$}
In this subsection, we show that the energy is cascaded away from $\{0\}$.
\begin{proposition} Suppose that $g$, $gp^2\in C([0,\infty):\mathfrak{D}([0,\infty]))$,  $g\ge0$ is a non-trivial weak solution  in the sense of Definition \ref{Def:WeakSolution} and
$$\int_{\{0\}} g_0(p)d\mu(p) =0.$$
For all $\epsilon\in(0,1)$, there exists $R_\epsilon>0$ such that
$$\int_{\left[0, R_\epsilon \right)}g(t,p)p^2d \mu(p) \ \le \ \epsilon E\text{ for all }t\in[0,\infty).$$

\end{proposition}
\begin{proof}
Let $R_1$ be a constant satisfying $\int_{R_1}^{\infty} g_0(p)|p|^2 d\mu(p) \ \ge \left(1-\frac{\epsilon}{2}\right)E.$
By Lemma \ref{Propo:Tightness}, it follows that
\begin{equation*}
\int_{\left[R_1 \frac{\epsilon}{2}, \infty\right]}g(t,p)p^2d \mu(p)\geq\left(1-\frac{\epsilon}{2}\right)\int_{[R_1,\infty]}g(0,p)p^2d\mu(p)\text{ for all }t\in[0,\infty),
\end{equation*}
which implies
\begin{equation*}
\int_{\left[R_1 \frac{\epsilon}{2}, \infty\right]}g(t,p)p^2d \mu(p)\geq\left(1-\frac{\epsilon}{2}\right)^2 E\text{ for all }t\in[0,\infty),
\end{equation*}
Since
$\left(1-\frac{\epsilon}{2}\right)^2 \ge 1-\epsilon,$
we have
\begin{equation*}
\int_{\left[R_1 \frac{\epsilon}{2}, \infty\right]}g(t,p)p^2d \mu(p)\geq\left(1-{\epsilon}\right) E\text{ for all }t\in[0,\infty).
\end{equation*}
Finally, choose $R_\epsilon= R_1 \frac{\epsilon}{2}$, and by the conservation of energy
$\int_{\left[0, \infty\right]}g(t,p)p^2d \mu(p) \ = \ E,$
we find
$$\int_{\left[0, R_\epsilon \right)}g(t,p)p^2d \mu(p) \ \le \ \epsilon E\text{ for all }t\in[0,\infty).$$
\end{proof}

\subsection{Positivity of the energy at $\{\infty\}$ as time evolves}\label{Sec:Mass0}
{The main result of this section is the following Proposition \ref{Propo:Mass0}, which shows that the energy at  $\{\infty\}$ is indeed strictly increasing, leading to the energy cascade phenomenon. Notice that similar phenomena have also been observed  by Lu \cite{Lu:2013:TBE}  in the context of Bose-Einstein Condensates and later by Kierkels and Vel\'azquez \cite{kierkels2015transfer} in the context of the nonlinear Schr\"odinger equation for the mass at $\{0\}$. In these works, it is proved that the mass at $\{0\}$ is also strictly increasing, leading to the condensation phenomenon. In both cases, the similarity is the fact that there is a strictly increasing accumulation of mass/energy towards the singular points $\{\infty\}$ and $\{0\}$.
}
\begin{proposition}\label{Propo:Mass0}
Given any nontrivial weak solution $g$ in the sense of Definition \ref{Def:WeakSolution}   such that $g|p|^2\in C([0,\infty):\mathfrak{D}([0,\infty]))$,  $g\ge0$. For any ${t}_0\in[0,\infty)$, there exists $t'>0$ such that the following is true
\begin{equation}\label{Theorem:Existence:3}
\int_{\{\infty\}}g(t',p)|p|^2d \mu(p)>\int_{\{\infty\}}g({t}_0,p)|p|^2d \mu(p),
\end{equation}
which means any nontrivial weak solution $g|p|^2\in C([0,\infty):\mathfrak{D}([0,\infty]))$,  $g\ge0$ to \eqref{WeakTurbulenceEnergy} has the following property  $$\exists t_*>0, \mbox{ such that }\int_{\{\infty\}}g(t_*,p)|p|^2d \mu(p)>0.$$ 
\end{proposition}

\subsubsection{A lower bound for the energy}\label{Sec:LowerBound}

\begin{lemma}\label{Lemma:LowerBound}
Suppose that $g$  is a nontrivial weak solution in the sense of Definition \ref{Def:WeakSolution}  such that  $gp^2\in C([0,\infty):\mathfrak{D}([0,\infty]))$,  $g\ge0$. There then exist constants $R,T\in(0,\infty)$, depending only on the initial datum $g_0p^2\in\mathfrak{D}([0,\infty])$, so that
\begin{equation}\label{Lemma:LowerBound:1}
\int_{[r,\infty]}g(t,p)p^2d \mu(p)\gtrsim rt\text{ for all }r\in[0,R]\text{ and all }t\in[0,T].
\end{equation}
\end{lemma}
\begin{proof}
Since we are interested only in nontrivial solutions, we can define $$\int_{(0,\infty)}g(0,p)p^2d \mu(p)=:E\in(0,\infty),$$
then by the conservation of energy
$\int_{(0,\infty)}g(t,p)p^2d \mu(p) \le E,$
for all $t\ge 0$.

Set $R_1,R_2\in(0,\infty)$ such that $$\int_{[R_1,R_2]}g(0,p)p^2d \mu(p)\geq\frac13E.$$
Let us now choose a test function $\phi\in \mathcal{V}$ satisfying $\|\phi\|_{\mathcal{V}}\leq 1$ and $\phi\equiv 1$ on $[R_1,R_2]$. Using then $\phi$ in \eqref{WeakTurbulenceInitial}, we obtain the following estimate $$\int_{[R_1,R_2]}g(t,p)p^2d \mu(p)\geq\frac13E-C E^2 t,$$ where $C$ depends only on $R_1$ and $R_2$. 

Therefore, for $T\in(0,\infty)$ small enough, 
\begin{equation}\label{Lemma:LowerBound:E1}
\int_{[R_1,R_2]}g(t,p)p^2d \mu(p)\geq\tfrac14E\text{ for all }t\in[0,T].
\end{equation}

Choose $r\in(0,R_1]$. Applying then Lemma \ref{lemma:KeyLemma}, we obtain  for all $t\in[0,\infty)$
\begin{equation}\label{Lemma:LowerBound:E2}
\begin{aligned}
&\int_{[r,\infty]}g(t,p)p^2d \mu(p)\gtrsim \int_0^t\iint_{[0,r]^2}g(s,p_1)p_1^2g(s,p_2)p_2^2\sqrt{p_1p_2}\left(1-\frac{r}{p_1+p_2}\right)_+d \mu(p_1)d \mu(p_2)d s\\
&\indent \gtrsim \int_0^t\iint_{\left[\frac{2r}{3},r\right]^2}g(s,p_1)p_1^2g(s,p_2)p_2^2\sqrt{p_1p_2}\left(1-\frac{r}{p_1+p_2}\right)_+d \mu(p_1)d \mu(p_2)d s.
\end{aligned}
\end{equation}

Note that $\frac14\leq\left(1-\frac{r}{p_1+p_2}\right)_+\leq\frac12$ for $p_1,p_2\in\left[\frac{2r}{3},r\right]$. Moreover, $[p_1p_2]^{\frac12}\geq \frac23r$ for $p_1,p_2\in\left[\frac{2r}{3},r\right]$. 

Plugging these estimates into the right hand side of \eqref{Lemma:LowerBound:E2}
\begin{equation}\label{Lemma:LowerBound:E3}
\begin{aligned}
&\int_{[r,\infty]}g(t,p)p^2d \mu(p) \gtrsim r\int_0^t\left[\iint_{\left[\frac{2r}{3},r\right]}g(s,p)p_1^2d \mu(p)\right]^2d s.
\end{aligned}\end{equation}

Now, applying \eqref{Lemma:LowerBound:E3} for $r=R_0\frac{3^{n}}{2^{n}}$, $n>1$, $0<R_0\le R_1$, we get

\begin{equation}\label{Lemma:LowerBound:E4}
\begin{aligned}
&\left(\frac23\right)^n\int_{\left[\frac{3^{n}}{2^{n}}R_0,\infty\right]}g(t,p)p^2d \mu(p) \gtrsim R_0\int_0^t\left[\iint_{\left[\frac{3^{n-1}}{2^{n-1}}R_0,\frac{3^{n}}{2^{n}}R_0\right]}g(s,p)p^2d \mu(p)\right]^2d s.
\end{aligned}\end{equation}

Let $N$ be an integer such that $[R_1,R_2]\subset \left[R_0,\frac{3^{N}}{2^{N}}R_0\right]$. We take the sum of \eqref{Lemma:LowerBound:E4} from $n=1$ to $n=N$, employ the Cauchy-Schwarz inequality and use the fact that $\left[R_0\frac{3^{n}}{2^{n}},\infty\right]\subset [R_0,\infty],$ to find
\begin{equation}\label{Lemma:LowerBound:E5}
\begin{aligned}
\int_{[R_0,\infty]}g(t,p)p^2d \mu(p) \ \gtrsim & \ R_0\int_0^t\left[\iint_{\left[R_0,\frac{3^{N}}{2^{N}}R_0\right]}g(s,p)p^2d \mu(p)\right]^2d s\\
 \ \gtrsim & \ R_0\int_0^t\left[\iint_{\left[R_1,R_2\right]}g(s,p)p^2d \mu(p)\right]^2d s.
\end{aligned}\end{equation}
which, together with \eqref{Lemma:LowerBound:E1}, implies 

\begin{equation}\label{Lemma:LowerBound:E5aa}
\begin{aligned}
\int_{[R_0,\infty]}g(t,p)p^2d \mu(p) \ \gtrsim & \ R_0E^2t \ \ \  \forall t\in[0,T].
\end{aligned}\end{equation}
and we get the conclusion of the lemma. 
\end{proof}

\begin{lemma}\label{Lemma:LowerBound2}
Suppose that  $g$, $gp^2\in C([0,\infty):\mathfrak{D}([0,\infty]))$,  $g\ge0$ is a non-trivial weak solution the sense of Definition \ref{Def:WeakSolution}. Then exist constants $R_*\in(0,\infty)$, depending only on $g(0,\cdot)$, and some $T\in(0,\infty)$ such that
\begin{equation}\label{LemmaLowerBound2:1a}
\int_{[r,\infty]}g(t,p)p^2d \mu(p) \gtrsim
 T\,r\text{ for all }r\in[0,R_*]\text{ and all }t\in[T,\infty).
\end{equation}
\end{lemma}

\begin{proof}
Without loss of generality we consider only nontrivial solutions. By Lemma \ref{Lemma:LowerBound} there  exist constants $R,T\in(0,\infty)$, depending only on $g(0,\cdot)$, such that \begin{equation}\label{Lemma:LowerBound:1}
\int_{[r,\infty]}g(t,p)p^2d \mu(p)\gtrsim rt\text{ for all }r\in[0,R]\text{ and all }t\in[0,T].
\end{equation} 
Now choose applying Proposition \ref{Propo:Tightness} with $\rho=\frac12$ implies that
\begin{equation*}\label{Propo:Tightness:1}
\int_{\left[\frac{r}{2}, \infty\right]}g(t+T,p)p^2d \mu(p)\geq \frac12\int_{[r,\infty]}g(T,p)p^2d \mu(p)\text{ for all }t\in[0,\infty),
\end{equation*}
which implies
\begin{equation*}\label{Lemma:LowerBound2:1}
\int_{\left[\frac{r}{2}, \infty\right]}g(t,p)p^2d \mu(p) \gtrsim
 T\,r\text{ for all }r\in[0,R]\text{ and all }t\in[T,\infty).
\end{equation*}
Hence \eqref{LemmaLowerBound2:1a} holds for $R_*=R/2$.
\end{proof}
\subsubsection{Proof of Proposition \ref{Propo:Mass0}}\label{Sec:AverageMass}

We will prove that any nontrivial weak solution  $g$, $gp^2\in C([0,\infty):\mathfrak{D}([0,\infty]))$,  $g\ge0$  to \eqref{WeakTurbulenceEnergy} has the property that $$\exists t_*>0, \mbox{ such that }\int_{\{\infty\}}g(t_*,p)|p|^2d \mu(p)>0.$$ 

According to \eqref{LemmaLowerBound2:1a}, there exist $R_*$ and $\mathcal{M}>0$ such that
\begin{equation*}
\int_{[\rho,\infty]}g(t,p)p^2d \mu(p) \geq
c \mathcal{M}\,\rho\text{ for all }\rho\in[0,R_*]\text{ and all }t\in[\mathcal{M},\infty),
\end{equation*}
for some universal constant $c>0$, that we suppose to be $1$ for the sake of simplicity.

Let $\vartheta>1$ and suppose  there exists $N>0$ such that
\begin{equation}\label{Lemma:LowerBound3:E1aa}
\int_{[\rho,\vartheta^N  \rho]}g(t,p)p^2d \mu(p)>\frac{\mathcal{M}\rho}{2}\text{ for all }\rho\in[0,R_*]\text{ and all }t\in[\mathcal{M},\infty).
\end{equation}

Let $n$ be an integer and $R_0<\rho$ such that ${[\rho,\vartheta^N  \rho]}\subset \left[R_0,\frac{3^{n}}{2^{n}}R_0\right]$, we recall from \eqref{Lemma:LowerBound:E5} that
\begin{equation*}
\begin{aligned}
\int_{[R_0,\infty]}g(t,p)p^2d \mu(p) \ \gtrsim & \ R_0\int_0^t\left[\iint_{\left[R_0,\frac{3^{n}}{2^{n}}R_0\right]}g(s,p)p^2d \mu(p)\right]^2d s\\
 \ \gtrsim & \ R_0\int_0^t\left[\iint_{\left[\rho,\vartheta^N  \rho\right]}g(s,p)p^2d \mu(p)\right]^2d s\
 \ \gtrsim  \ \frac{\mathcal{M}^2\rho^2 R_0t}{4}.
\end{aligned}\end{equation*}
which means
\begin{equation*}
\begin{aligned}
\int_{[0,\infty]}g(0,p)p^2d \mu(p) \ = & \int_{[0,\infty]}g(t,p)p^2d \mu(p)
 \ \gtrsim \ \int_{[R_0,\infty]}g(t,p)p^2d \mu(p)
 \ \gtrsim  \ \frac{\mathcal{M}^2\rho^2 R_0t}{4} \ \to \infty 
\end{aligned}\end{equation*}
as $t$ goes to $\infty$. This leads to a contradiction since the left hand side is a constant and the right hand side tends to infinity as $t$ tends to infinity.

Therefore   \eqref{Lemma:LowerBound3:E1aa} is false. Then there must exist $\rho\in[0,R_*]$ and $t_*\in[\mathcal M,\infty)$ such that
\begin{equation}\nonumber
\int_{[\rho,\vartheta^N\rho]}g(t_*,p)p^2 d \mu(p)\leq\frac{\mathcal M \rho}{2},
\end{equation}
for all $N>0$. 

Since $\vartheta>1$, 
$\lim_{N\to \infty}\vartheta^N\rho=\infty,$
we deduce a consequence of the above 
\begin{equation}\nonumber
\int_{\{\infty\}}g(t_*,p)p^2 d \mu(p)\ge \frac{\mathcal M \rho}{2}>0.
\end{equation}

Now, we will prove \eqref{Theorem:Existence:3}. Suppose the contrary that for all $t>t_1$
$$\int_{\{\infty\}}g(t,p)p^2d\mu(p) \ = \ \int_{\{\infty\}}g(t_1,p)p^2d\mu(p).$$

We can use exactly the same argument as before, but Lemma \ref{lemma:KeyLemma} is then modified  by
\begin{equation}\label{Lemma:LowerBound3:E1aa1}
\int_{[r,\infty)}g(t,p)p^2d \mu(p)\geq2\int_{t_1}^t\iint_{[0,r]^2}{g(s,p_1)p_1^2g(s,p_2)p_2^2}{[p_1p_2]^\frac12}\left(1-\frac{r}{p_1+p_2}\right)_+d \mu(p_1)d \mu(p_2)\ d s,
\end{equation}
in which the initial condition is chosen at $t_1$. 

This inequality can be proved as follows. Similar with \eqref{lemma:KeyLemma:E1}, one has

\begin{equation}\label{lemma:KeyLemma:E1new}
\begin{aligned}
& \int_{[0,\infty]}g(t,p)p^2\phi(p)dp - \int_{[0,\infty]}g(t_1,p)p^2\phi(p)d\mu(p)  \\
 =& \ 2\int_{t_1}^t\int_{p_1>p_2\ge0}{|p_1||p_2|}g(s,p_1)g(s,p_2)\mathfrak{H}^1_\phi(p_1,p_2){d}\mu(p_{1})\,{d}\mu(p_2)ds\,\\
 &  + \int_{t_1}^t\int_{p_1=p_2\ge0}{|p_1||p_2|}g(s,p_1)g(s,p_2)\mathfrak{H}^2_\phi(p_1,p_2) {d}\mu(p_{1})\,{d}\mu(p_2)ds\,.
\end{aligned}
\end{equation}
which, by the fact that $g\ge 0$, yields
\begin{equation}\label{lemma:KeyLemma:E2new}
\begin{aligned}
& \int_{[0,\infty]}g(t,p)p^2\phi(p)d\mu(p) - \int_{\{p=\infty\}}g(t_1,p)p^2\phi(p)d\mu(p) \\ \ \ge & \ 2\int_{t_1}^t\int_{p_1>p_2\ge0}{|p_1||p_2|}g(s,p_1)g(s,p_2)\mathfrak{H}^1_\phi(p_1,p_2){d}\mu(p_{1})\,{d}\mu(p_2)ds\,\\
 &  + \int_{t_1}^t\int_{p_1=p_2\ge0}{|p_1||p_2|}g(s,p_1)g(s,p_2)\mathfrak{H}^2_\phi(p_1,p_2) {d}\mu(p_{1})\,{d}\mu(p_2)ds\,.
\end{aligned}
\end{equation}

Notice that $\int_{\{p=\infty\}}g(t_1,p)p^2\phi(p)d\mu(p)= \int_{\{p=\infty\}}g(t_1,p)p^2d\mu(p)= \int_{\{p=\infty\}}g(t,p)p^2d\mu(p).$
Therefore, similar as \eqref{lemma:KeyLemma:E2b}, we also have
\begin{equation}\label{lemma:KeyLemma:E2bnewa}
\begin{aligned}
& \int_{[0,\infty]}g(t,p)p^2\phi(p)d\mu(p) \ - \ \int_{\{p=\infty\}}g(t_1,p)p^2d\mu(p)\\
 \ge & \ \int_{t_1}^t\iint_{[0,r]^2}{g(s,p_1)p_1g(s,p_2)p_2}(p_1^3+3p_1p_2^2)\phi(p_1+p_2)d \mu(p_1)d \mu(p_2)d s,
\end{aligned}
\end{equation}
which, by the Cauchy-Schwarz inequality $p_1+p_2\ge 2\sqrt{p_1p_2}$, leads to
\begin{equation}\label{lemma:KeyLemma:E2bnew}
\begin{aligned}
& \int_{[0,\infty]}g(t,p)p^2\phi(p)d\mu(p) \ - \ \int_{\{p=\infty\}}g(t_1,p)p^2d\mu(p)\\
 \ge & \ 2\int_{t_1}^t\iint_{[0,r]^2}{g(s,p_1)p_1^2g(s,p_2)p_2^2}(p_1p_2)^{
\frac12}\phi(p_1+p_2)d \mu(p_1)d \mu(p_2)d s,
\end{aligned}
\end{equation}

The above inequalities  yields
\begin{equation}\label{lemma:KeyLemma:E3new}
\int_{[r,\infty)}\,g(t,p)p^2d \mu(p)\geq 2\int_{t_1}^t\iint_{[0,r]^2}{g(s,p_1)p_1^2g(s,p_2)p_2^2}{(p_1p_2)^{
\frac12}}\phi(p_1+p_2)d \mu(p_1)d \mu(p_2)d s,
\end{equation}
which then implies \eqref{Lemma:LowerBound3:E1aa1}. 

The same argument of Lemma \ref{Lemma:LowerBound} can be applied to deduce that there exist $0<R_1<R_2$ satisfying for $T\in(0,\infty)$ small enough, 
\begin{equation}\label{Lemma:LowerBound:E1new}
\int_{[R_1,R_2]}g(s,p)p^2d \mu(p)\geq\tfrac14E\text{ for all }s\in[t_1,t_1+T]
\end{equation}
as well as 
\begin{equation}\label{Lemma:LowerBound:E4new}
\begin{aligned}
&\left(\frac23\right)^n\int_{\left[\frac{3^{n}}{2^{n}}R_0,\infty\right)}g(t,p)p^2d \mu(p) \gtrsim R_0\int_0^t\left[\iint_{\left[\frac{3^{n-1}}{2^{n-1}}R_0,\frac{3^{n}}{2^{n}}R_0\right]}g(s,p)p^2d \mu(p)\right]^2d s.
\end{aligned}\end{equation}

Letting $N$ be an integer such that $[R_1,R_2]\subset \left[R_0,\frac{3^{N}}{2^{N}}R_0\right]$,  taking the sum of \eqref{Lemma:LowerBound:E4new} from $n=1$ to $n=N$, employing the Cauchy-Schwarz inequality and using the fact that $\left[R_0\frac{3^{n}}{2^{n}},\infty\right)\subset [R_0,\infty),$ yields
\begin{equation}\label{Lemma:LowerBound:E5new}
\begin{aligned}
\int_{[R_0,\infty)}g(t,p)p^2d \mu(p) \ \gtrsim & \ R_0\int_{t_1}^t\left[\iint_{\left[R_0,\frac{3^{N}}{2^{N}}R_0\right]}g(s,p)p^2d \mu(p)\right]^2d s\\
 \ \gtrsim & \ R_0\int_{t_1}^t\left[\iint_{\left[R_1,R_2\right]}g(s,p)p^2d \mu(p)\right]^2d s.
\end{aligned}\end{equation}
 
That, with a similar argument with \eqref{LemmaLowerBound2:1a} also gives the existence of $R_*$ and $\mathcal{M}>0$ such that
\begin{equation*}
\int_{[\rho,\infty)}g(t,p)p^2d \mu(p) \ge
 C\mathcal{M}\,\rho\text{ for all }\rho\in[0,R_*]\text{ and all }t\in[t_1+\mathcal{M},\infty),
\end{equation*}
for some universal constant $C>0$.

Let us fix $\vartheta>1$, then there exists $N\in\mathbb{N}$ such that
\begin{equation}\label{Lemma:LowerBound3:E1aaa}
\int_{[\rho,\vartheta^N  \rho]}g(t,p)p^2d \mu(p)\ge \frac{C\mathcal{M}\rho}{2}\text{ for all }\rho\in[0,R_*]\text{ and all }t\in[t_1+\mathcal{M},\infty);
\end{equation} 
Otherwise, if  there exist $\rho\in[0,R_*]\text{ and  }t\in[t_1+\mathcal{M},\infty)$ such that
\begin{equation*}
\int_{[\rho,\vartheta^N  \rho]}g(t,p)p^2d \mu(p)< \frac{C\mathcal{M}\rho}{2};
\end{equation*}
for all $N\in\mathbb{N}$, then
\begin{equation*}
\int_{[\rho,\infty)}g(t,p)p^2d \mu(p) \le \frac{C\mathcal{M}\rho}{2},
\end{equation*}
which is a contradiction.

As   a consequence,  by \eqref{Lemma:LowerBound3:E1aaa} we deduce  that
\begin{equation*}
\begin{aligned}
\int_{[0,\infty)}g(t_1,p)p^2d \mu(p) \ = & \int_{[0,\infty)}g(t,p)p^2d \mu(p)
 \ \gtrsim  \ \frac{\mathcal{M}^2\rho^2 (t-t_1)}{4} \ \to \infty 
\end{aligned}\end{equation*}
as $t$ goes to $\infty$, which is also a contradiction.

\subsection{Rate of cascading the energy toward $\{\infty\}$}
The main result of this section is the following Proposition, which gives a rate of cascading the energy toward $\{\infty\}$.

\begin{proposition}\label{Propo:RateCascade}
Given any non-trivial weak solution $g$ in the sense of Definition \ref{Def:WeakSolution}   such that $g|p|^2\in C([0,\infty):\mathfrak{D}([0,\infty]))$,  $g\ge0$, and
$$\int_{\{0\}}p^2g_0(p)d \mu(p)=\int_{\{\infty\}}p^2g_0(p)d \mu(p)=0.$$
 There exist explicit constants $C_1$, $C_2$, $T^*$ depending on the initial condition $g_0$ such that   the following holds true
\begin{equation}\label{Propo:RateCascade:1}
\int_{\{\infty\}}g(t,p)|p|^2d \mu(p)\ \ge \ C_1 \ - \ \frac{C_2}{\sqrt{t}},
\end{equation}
for all $t>T^*$. 

That leads to the existence of an explicit $T^{**}>0$ such that 
\begin{equation}\label{Propo:RateCascade:2}
\int_{\{\infty\}}g(t,p)|p|^2d \mu(p)\ \ge\ \frac{C_1}{2},
\end{equation}
for all $t>T^{**}$. 

Moreover, there exists $R^*>0$ such that for any $r>R^*$, there exists an explicit $T_r>0$ and for all $t>T_r$ 
\begin{equation}\label{Propo:RateCascade:3}
\int_{[r,\infty]}g(t,p)|p|^2d \mu(p)\ \ge \ C_1.
\end{equation}
\end{proposition}
\subsubsection{An upper bound for the energy }\label{Sec:UpperBound}
The following lemma gives an estimate of the energy on any interval $[R,\infty)$. 
\begin{lemma}\label{Lemma:UpperBound}
Suppose that  $g$, $gp^2\in C([0,\infty):\mathfrak{D}([0,\infty]))$,  $g\ge0$ is a non-trivial weak solution  in the sense of Definition \ref{Def:WeakSolution}  and
$$\int_{\{0\}}p^2g_0(p)d \mu(p)=\int_{\{\infty\}}p^2g_0(p)d \mu(p)=0.$$ Then for all $R\in[0,\infty)$, all $t_1\in[0,\infty)$ and any $t_2\in[t_1,\infty)$, the following holds true
\begin{equation}\label{Lemma:UpperBound:1}
\int_{t_1}^{t_2}\int_{[R,\infty)}g(s,p)p^2d \mu(p)d s \lesssim \sqrt\frac{{(t_2-t_1)\|g_0p^2\|_1}}{{R}}.
\end{equation}
\end{lemma}
\begin{proof}
Choosing $t_1\in[0,\infty)$ and $t_2\in[t_1,\infty)$,  applying Lemma \ref{lemma:KeyLemma} and using the fact $\frac43 r \le p_1+p_2\leq 2{r}$ on $[\frac23r,r]\times[\frac23r,r]$, we obtain
\begin{equation}\label{Lemma:UpperBound:E1a}\begin{aligned}
\int_{[r,\infty]}g(t_2,p)p^2 d\mu(p) \ \gtrsim & \ \int_0^{t_2}\int_{[0,r]^2}g(s,p_1)p_1^2g(s,p_2)p_2^2 \sqrt{p_1p_2}d\mu(p_1)d\mu(p_2)ds.
\end{aligned}
\end{equation}
Based on the observation $\left[\frac23r,r\right]\subset [0,r]$, we obtain from the above inequality that
\begin{equation}\label{Lemma:UpperBound:E1b}\begin{aligned}
\int_{[r,\infty]}g(t_2,p)p^2 d\mu(p) 
\ \gtrsim & \ \int_{t_1}^{t_2}\int_{\left[\frac23r,r\right]^2}g(s,p_1)p_1^2g(s,p_2)p_2^2 \sqrt{p_1p_2}d\mu(p_1)d\mu(p_2)ds.
\end{aligned}
\end{equation}
Applying the inequality $\sqrt{p_1p_2}\ge \frac{2}{3}r$ into the above estimate, we deduce that
\begin{equation}\label{Lemma:UpperBound:E1}\begin{aligned}
\int_{[r,\infty]}g(t_2,p)p^2 d\mu(p) 
\ \gtrsim & \ \frac23r\int_{t_1}^{t_2}\left[\int_{\left[\frac23r,r\right]}g(s,p)p^2d\mu(p)\right]^2ds,
\end{aligned}
\end{equation}
It now follows from H\"older's inequality that
$$(t_2-t_1)\int_{t_1}^{t_2}\left[\int_{\left[\frac23r,r\right]}g(s,p)p^2d\mu(p)\right]^2ds \ \ge \ \left[\int_{t_1}^{t_2}\int_{\left[\frac23r,r\right]}g(s,p)p^2d\mu(p)ds\right]^2, $$
which, in combination with \eqref{Lemma:UpperBound:E1} yields
\begin{equation}\label{Lemma:UpperBound:E2}\begin{aligned}
\int_{[r,\infty]}g(t_2,p)p^2 d\mu(p) \ \gtrsim &  \ \frac{r}{t_2-t_1}\left[\int_{t_1}^{t_2}\int_{\left[\frac23r,r\right]}g(s,p)p^2d\mu(p)ds\right]^2.
\end{aligned}
\end{equation}
Since $\int_{[r,\infty]}g(t_2,p)p^2 d\mu(p) \ \le \|gp^2\|_{L^1},$
we deduce from \eqref{Lemma:UpperBound:E2} that 
\begin{equation}\label{Lemma:UpperBound:E3}
\|gp^2\|_{L^1} \ \gtrsim   \ \frac{r}{t_2-t_1}\left[\int_{t_1}^{t_2}\int_{\left[\frac23r,r\right]}g(s,p)p^2d\mu(p)ds\right]^2.
\end{equation}
Applying \eqref{Lemma:UpperBound:E3} for $r=\left(\frac32\right)^nR$, $(n\ge 1)$,  we obtain
\begin{equation}\label{Lemma:UpperBound:E4}\begin{aligned}
\left(\frac23\right)^n\|gp^2\|_{L^1} \ \gtrsim &  \ \frac{R}{t_2-t_1}\left[\int_{t_1}^{t_2}\int_{\left[\left(\frac32\right)^{n-1}R,\left(\frac32\right)^nR\right]}g(s,p)p^2d\mu(p)ds\right]^2.
\end{aligned}
\end{equation}
which, by taking the square roots of both sides, implies
\begin{equation}\label{Lemma:UpperBound:E4}\begin{aligned}
\left(\frac23\right)^{\frac{n}{2}}\sqrt{\|gp^2\|_{L^1}} \ \gtrsim &  \ \sqrt{\frac{R}{t_2-t_1}}\left[\int_{t_1}^{t_2}\int_{\left[\left(\frac32\right)^{n-1}R,\left(\frac32\right)^nR\right]}g(s,p)p^2d\mu(p)ds\right],
\end{aligned}
\end{equation}
Taking the sum of \eqref{Lemma:UpperBound:E4} for $n\ge 1$, 
\begin{equation}\label{Lemma:UpperBound:E5}\begin{aligned}
\sum_{n=1}^\infty\left(\frac23\right)^{\frac{n}{2}}\sqrt{\|gp^2\|_{L^1}}\ \gtrsim &  \ \sqrt{\frac{R}{t_2-t_1}}\sum_{n=1}^\infty\left[\int_{t_1}^{t_2}\int_{\left[\left(\frac32\right)^{n-1}R,\left(\frac32\right)^nR\right]}g(s,p)p^2d\mu(p)ds\right],
\end{aligned}
\end{equation}
which can be simplified into
\begin{equation}\label{Lemma:UpperBound:E5}
\sqrt{\|gp^2\|_{L^1}}\ \gtrsim   \ \sqrt{\frac{R}{t_2-t_1}}\left[\int_{t_1}^{t_2}\int_{[R,\infty)}g(s,p)p^2d\mu(p)ds\right],
\end{equation}

The estimate \eqref{Lemma:UpperBound:1} follows.
\end{proof}

\subsubsection{Proof of Proposition \ref{Propo:RateCascade}}\label{Sec:ProofMass0}
By Lemma \ref{Lemma:LowerBound2}, there are $r, T>0$ such that
\begin{equation*}
\int_{[r,\infty]}g(t,p)p^2 p \geq
\gamma:= c_0 T\,r\text{ for all }t\in[T,\infty),
\end{equation*}
for some constant $c_0$.

For any $\tau \in\left[T,\infty\right)$, 
\begin{equation}\label{Sec:ProofMass0:E2}
\int_{\tau}^{2\tau}\int_{[r,\infty]}g(s,p)p^2d \mu(p)d s\geq \tau{\gamma}.
\end{equation}
In addition, by  Lemma \ref{Lemma:UpperBound}, we obtain
\begin{equation}\label{Sec:ProofMass0:E3}
\int_{\tau}^{2\tau}\int_{[r,\infty)}g(s,p)p^2d \mu(p)d s\leq C\sqrt{\tau\|g(0,p)p^2\|_{L^1}}\sqrt{\frac{1}{r}}\text{ for all }r\in[0,\infty),
\end{equation}
for some universal constant $C>0$.
\\
Subtracting \eqref{Sec:ProofMass0:E3} and \eqref{Sec:ProofMass0:E2}, we get
\begin{equation}\label{Sec:ProofMass0:E3a}
\int_{\tau}^{2\tau}\int_{\{\infty\}}g(s,p)p^2d \mu(p)d s\geq \tau{\gamma} \ - \ C\sqrt{\tau\|g(0,p)p^2\|_{L^1}}\sqrt{\frac{1}{r}}.
\end{equation}
Using Proposition \ref{Propo:Monoton}, we find
\begin{equation}\label{Sec:ProofMass0:E4}
\int_{\tau}^{2\tau}\int_{\{\infty\}}g(s,p)p^2d \mu(p)d s \ \leq \ \tau \int_{\{\infty\}}g(2\tau,p)p^2d \mu(p).
\end{equation}
Putting together \eqref{Sec:ProofMass0:E3a} and \eqref{Sec:ProofMass0:E4} yields
\begin{equation}\label{Sec:ProofMass0:E5}
 \int_{\{\infty\}}g(2\tau,p)p^2d \mu(p) \ \geq \ {\gamma} \ - \ \frac{C}{\sqrt\tau}\sqrt{\|g(0,p)p^2\|_{L^1}}\sqrt{\frac{1}{r}}
\end{equation}
which  leads to 
\begin{equation}\label{Sec:ProofMass0:E5}
 \int_{\{\infty\}}g(2\tau,p)p^2d \mu(p) \ \geq \ \tfrac{\gamma}{2},
\end{equation}
when $\tau$ is large enough and the result follows.

\subsection{Full cascade of the energy: the convergence to the Dirac function $\delta_{\{p=\infty\}}$}\label{Sec:Longtime}
\begin{proposition}[Long time behavior]\label{Propo:Longtime}
Suppose that $gp^2\in C([0,\infty):\mathfrak{D}([0,\infty]))$,  $g\ge0$  and $g$ is a non-trivial weak solution  in the sense of Definition \ref{Def:WeakSolution}  and
$$\int_{\{0\}}p^2g_0(p)d \mu(p)=\int_{\{\infty\}}p^2g_0(p)d \mu(p)=0.$$  Then $$g(t,p)p^2\stackrel{\ast}{\rightharpoonup} \|g(0,p)p^2\|_{L^1}\delta_{\{p=\infty\}}$$ as $t\rightarrow\infty$.
\end{proposition}
\begin{proof}
Let us consider only nontrivial solutions. Suppose that $\phi$ is a continuous function, belongs to the class $\mathfrak{L} = \{\phi_r\}$, constructed in Proposition \ref{Propo:L1L2Positivity}.  The mapping $t\mapsto \mathfrak{W}_\phi(t):=\int_{[0,\infty]}\phi(p)g(t,p)p^2d \mu(p)$ is then nondecreasing. In addition, 
$\mathfrak{W}_\phi(t) \le \|g_0p^2\|_{L^1}\in(0,\infty).$
Therefore, there exists a limit $\lim_{t\rightarrow\infty}\mathfrak{W}_\phi(t)=:\mathfrak{E}\in[0,E].$
Using the fact that $\phi<1$ on $[0,\infty)$, we observe that if $\mathfrak{E}=E$,
then
$\lim_{t\to\infty}\sup_{\phi \in \mathfrak{L}}\mathfrak{W}_\phi = E,$
which implies
$g(t,p)p^2\stackrel{\ast}{\rightharpoonup} \|g(0,p)p^2\|_{L^1}\delta_{\{p=\infty\}}$ as $t\rightarrow\infty$.
Suppose the contrary that $E-\mathfrak{E}=:\varepsilon>0$. We will prove that there exist $0<R_1<R_2<\infty$ such that
\begin{equation}\label{Propo:Longtime:E1}
\int_{[R_1,R_2)}g(t,p)p^2 d \mu(p)\ge \frac{\varepsilon}{4}\text{ for all }t\in[0,\infty).
\end{equation}
To this end,  let us consider the constants $\delta_1\in\left[\frac34,1\right)$ and $\frac{\mathfrak{E}}{E}(1-\delta_1)<\delta_2<1-\delta_1.$ For any $R_2\in[0,\infty)$, $\phi\in\mathfrak{L}$ it follows
\begin{equation}\label{Propo:Longtime:E2}
\int_{[R_2,\infty]}g(t,p)p^2d \mu(p)\leq\int_{[0,\infty]}\frac{\phi(p)}{\phi(R_2)}g(t,p)p^2d \mu(p)\leq\frac{\mathfrak{E}}{\phi(R_2)}\leq\frac{\delta_1 \mathfrak{E}+\delta_2 E}{\phi(R_2)}\text{ for all }t\in[0,\infty),
\end{equation}
which, by choosing $R_2\in(0,\phi^{-1}(\delta_1+\delta_2)]$, leads to
\begin{equation}\label{Propo:Longtime:E3}
\int_{[0,R_2)}g(t,p)p^2d \mu(p)\geq E-\frac{\delta_1 \mathfrak{E}+\delta_2 E}{\delta_1+\delta_2}>\frac{3\varepsilon}{4}\text{ for all }t\in[0,\infty).
\end{equation}

Let $R'\in(0,\infty)$ such that
\begin{equation}\label{Propo:Longtime:E4}
\int_{[R',\infty]}g(0,p)p^2 d \mu(p)\geq E-\frac{\varepsilon}{8},
\end{equation}
which, by applying Proposition \ref{Propo:Tightness} with $\rho=\frac{\varepsilon}{8E-\varepsilon}$, yields
\begin{equation}\label{Propo:Longtime:E5a}
\int_{[{R'}{\rho}, \infty]}g(t,p)p^2d \mu(p)\ge E-\frac{\varepsilon}{4}\text{ for all }t\in[0,\infty).
\end{equation}
That  means
\begin{equation}\label{Propo:Longtime:E5}
\int_{[0,{R'}{\rho})}g(t,p)p^2d\mu( p) <\frac{\varepsilon}{4}\text{ for all }t\in[0,\infty),
\end{equation}

Combining \eqref{Propo:Longtime:E3} and \eqref{Propo:Longtime:E5}, we obtain \eqref{Propo:Longtime:E1} for $R_1:={R'}{\rho}.$
Let $N$ be an integer and $R_0<R_1$ such that $[R_1,R_2]\subset \left[R_0,\frac{3^{N}}{2^{N}}R_0\right]$, we recall from \eqref{Lemma:LowerBound:E5} that
\begin{equation*}
\begin{aligned}
\int_{[R_0,\infty]}g(t,p)p^2d \mu(p) \ \gtrsim & \ R_0\int_0^t\left[\iint_{\left[R_0,\frac{3^{N}}{2^{N}}R_0\right]}g(s,p)p^2d \mu(p)\right]^2d s\\
 \ \gtrsim & \ R_0\int_0^t\left[\iint_{\left[R_1,R_2\right]}g(s,p)p^2d \mu(p)\right]^2d s\
  \gtrsim \ \frac{\varepsilon^2 R_0t}{16}.
\end{aligned}\end{equation*}
The above inequality leads to
\begin{equation*}
\begin{aligned}
\int_{[0,\infty]}g(0,p)p^2d \mu(p) \ = & \int_{[0,\infty]}g(t,p)p^2d \mu(p)
 \ \gtrsim  \int_{[R_0,\infty]}g(t,p)p^2d \mu(p)
 \ \gtrsim  \ \frac{\varepsilon^2 R_0t}{16} \ \to \infty 
\end{aligned}\end{equation*}
as $t$ to $\infty$. 
This is a contradiction.  Therefore $\mathfrak{E}=E$ and that implies the conclusion of the lemma. 

\end{proof}

\subsection{Proof of Theorem \ref{Theorem:Existence}}\label{Sec:Theorem1}
The proof follows from the previous propositions.

{\bf Acknowledgements.} 
A. Soffer is partially supported by
NSF grant DMS 1600749  and NSFC 11671163.
M.-B. Tran is partially supported by NSF Grant DMS-1814149 and NSF Grant DMS-1854453. The authors would like to thank Prof. A. Newell, Prof. A. Aceves,  Prof R. Pego, Prof. B. Rumpf, Prof H. Spohn for fruitful discussions on the topic. 

\def\cprime{$'$}

\end{document}